%% file: ms.tex
\documentclass[sigplan,screen]{acmart}
\AtBeginDocument{%
	}

\copyrightyear{2023}
\acmYear{2023} 
\setcopyright{rightsretained}
\acmConference[PPoPP '23]{The 28th ACM SIGPLAN Annual Symposium on
	Principles and Practice of Parallel Programming}{February 25-March 1,
	2023}{Montreal, QC, Canada}
\acmBooktitle{The 28th ACM SIGPLAN Annual Symposium on Principles
	and Practice of Parallel Programming (PPoPP '23), February 25-March 1,
	2023, Montreal, QC, Canada}
\acmDOI{10.1145/3572848.3577508}
\acmISBN{979-8-4007-0015-6/23/02}

\usepackage[utf8]{inputenc} 
\usepackage[T1]{fontenc}    
\usepackage{hyperref}       
\usepackage{url}            
\usepackage{booktabs}       
\usepackage{amsfonts}       
\usepackage{nicefrac}       
\usepackage{microtype}      
\usepackage{cleveref}       
\usepackage{lipsum}         
\usepackage{graphicx}
\usepackage{natbib}
\usepackage{xcolor}
\usepackage{caption}
\usepackage{subcaption}
\usepackage{doi}
\usepackage{listings, fancyvrb}
\usepackage{amsthm}
\usepackage{thmtools,thm-restate}
\usepackage{enumitem} 

\makeatletter
\let\bigtimes\@undefined
\makeatother
\usepackage{mathabx}

\input{macros}

\arxcam{
	\settopmatter{printacmref=false}
	\setcopyright{none}
	\renewcommand\footnotetextcopyrightpermission[1]{}
	\pagestyle{plain}
}{
	\usepackage[compact]{titlesec}
}

\begin{document}

\title[Practically and Theoretically Efficient Garbage Collection for Multiversioning]{Practically and Theoretically Efficient\\Garbage Collection for Multiversioning}

\author{Yuanhao Wei}
\email{yuanhao1@cs.cmu.edu}
\orcid{0000-0002-5176-0961}
\affiliation{
	\institution{Carnegie Mellon University, USA}
	\country{}}

\author{Guy E. Blelloch}
\email{guyb@cs.cmu.edu}
\orcid{0000-0003-0224-9187}
\affiliation{
	\institution{Carnegie Mellon University, USA}
	\country{}}

\author{Panagiota Fatourou}
\email{faturu@csd.uoc.gr}
\orcid{0000-0002-6265-6895}
\affiliation{
	\institution{FORTH ICS and University of Crete, Greece}
	\country{}}

\author{Eric Ruppert}
\email{ruppert@eecs.yorku.ca}
\orcid{0000-0001-5613-8701}
\affiliation{\institution{York University, Canada}
		\country{}}

\input{abstract}


\arxcam{
}{}


\maketitle

\arxcam{
	\pagestyle{plain}
}{}

\input{intro2}
\input{related}

\input{preliminaries}

\input{doubly-linked-list}
\input{singly-linked-list}

\input{experiments}

\input{future-work}
\begin{acks}
We thank the anonymous referees for their helpful comments.
This research was supported by NSERC; NSF grants CCF-1901381, CCF-1910030, CCF-1919223, and CCF-2119069; and
\y{the Hellenic Foundation
for Research and Innovation (HFRI) under the ``Second Call for HFRI Research
Projects to support Faculty members and researchers'' (project number: 3684).}
\end{acks}

\bibliographystyle{ACM-Reference-Format}
\bibliography{biblio}

\arxcam{\input{appendix}}{
}

\end{document}

%% file: macros.tex
\def\draft{0}  
\if\draft 1
\newcommand{\mycomment}[3]{{\color{#2}{[\textbf{#1: #3}]}}}
\newcommand{\myedit}[2]{{\color{#1}{#2}}\normalcolor}

\newcommand{\here}[1]{{\bf [[[#1]]]}}
\else
\newcommand{\mycomment}[3]{}

\newcommand{\myedit}[2]{#2}
\newcommand{\here}[1]{}  
\fi

\newcommand{\arxcam}[2]{#1}  

\newcommand{\bbf}{{\sf BBF+}}
\newcommand{\TreeDL}{{\sf TreeDL}}
\newcommand{\RTDLGC}{{\sf DL-RT}}
\newcommand{\DL}{{\sf PDL}}
\newcommand{\RTSLGC}{{\sf SL-RT}}
\newcommand{\SL}{{\sf SSL}}
\newcommand{\STEAMLF}{{\sf STEAM+LF}}
\newcommand{\EBR}{{\sf EBR}}

\newcommand{\Eric}[1]{\mycomment{Eric}{magenta}{#1}}

\newcommand{\Hao}[1]{\mycomment{Hao}{brown}{#1}}

\newcommand{\Youla}[1]{\mycomment{Youla}{cyan}{#1}}

\newcommand{\er}[1]{\myedit{magenta}{#1}}
\newcommand{\y}[1]{\myedit{cyan}{#1}}
\newcommand{\hedit}[1]{\myedit{brown}{#1}}

\newcommand{\var}[1]{\texttt{#1}}

\newcommand{\vw}{\var{w}}
\newcommand{\vx}{\var{x}}
\newcommand{\vy}{\var{y}}
\newcommand{\vz}{\var{z}}
\newcommand{\tryAppend}{\var{tryAppend}}
\newcommand{\remove}{\var{remove}}
\newcommand{\search}{\var{search}}

\newcommand{\compact}{\var{compact}}

\newcommand{\lreach}{\leftsquigarrow}

\newcommand{\op}[1]{\co{#1}} 
\newcommand{\deprecate}{{\tt deprecate}}

\newcommand{\announce}{{\tt announce}}
\newcommand{\unannounce}{{\tt unannounce}}

\newcommand{\rtx}{rtx}
\newcommand{\rtxs}{rtxs}

\newcommand{\RT}{{\sc RangeTracker}}
\newcommand{\RTO}{{\sc RT}}

\makeatletter
\lst@Key{countblanklines}{true}[t]%
{\lstKV@SetIf{#1}\lst@ifcountblanklines}

\lst@AddToHook{OnEmptyLine}{%
	\lst@ifnumberblanklines\else%
	\lst@ifcountblanklines\else%
	\advance\c@lstnumber-\@ne\relax%
	\fi%
	\fi}
\makeatother

\newtheorem{result}{XXX}
\newtheorem{theorem}[result]{Theorem}

\newtheorem{lemma}[result]{Lemma}

\newtheorem{proposition}[result]{Proposition}
\newtheorem{invariant}[result]{Invariant}

\newtheorem{claim}{Claim}[theorem]

\newcommand{\ignore}[1]{}

%% file: abstract.tex

\begin{abstract}
Multiversioning is widely used in databases, transactional memory, and concurrent data structures.
It can be used to support read-only transactions that appear atomic in the presence of concurrent update operations.
Any system that maintains multiple versions of each object needs a way of efficiently reclaiming them. 
We experimentally compare various existing reclamation techniques by applying them to a multiversion tree and a multiversion hash table.

Using insights from these experiments, we develop 
two new multiversion garbage collection (MVGC) techniques.
These techniques use two novel concurrent version list data structures.
Our experimental evaluation shows that our fast\-est technique is competitive with the fastest existing MVGC techniques, while using significantly less space on 
some workloads.
Our new techniques  provide strong theoretical bounds, especially on space usage.
These bounds ensure that the schemes have 
consistent performance, avoiding the very high worst-case space usage 
of other techniques.

\end{abstract}

%% file: intro2.tex

\section{Introduction}
\label{sec:intro}

Multiversion concurrency control is used widely in database systems
\cite{Reed78,BG83,papadimitriou1984concurrency,HANAgc,neumann2015fast,Wu17},
transactional
memory~\cite{perelman2010maintaining,FC11,Perelman11,Kumar14,Keidar2015},
and shared data structures~\cite{BBB+20,FPR19,NHP22,WBBFRS21a}, mainly
for supporting {\it read-only transactions} (abbreviated as {\it \rtxs}),
including complex multi-point queries that appear atomic, while update
operations are executed concurrently. 
This is usually achieved by maintaining a list of previous versions
for each object,
sorted by timestamps indicating when the object was updated.
An \rtx\ gets a
timestamp $t$ and 
traverses objects' {\em version lists},
to read the latest version at or before~$t$.

Multiversion garbage collection (MVGC) is the problem of removing and
collecting obsolete versions.  
Suppose two consecutive versions in the version list have timestamps $t_1$ and $t_2$.
The version with timestamp $t_1$ is \emph{needed} if there is an ongoing \rtx\ whose
timestamp $t$ satisfies $t_1\leq t < t_2$, and \emph{obsolete} otherwise.
The latest version is always needed.
Efficient MVGC is an important problem in multiversioning
systems~\cite{lu2013generic,bottcher2019scalable,HANAgc,BBF+21,FC11,neumann2015fast,kim2019mvrlu,NHP22,WBBFRS21a}, 
since keeping obsolete versions can lead to excessive space usage.

The simplest approach to MVGC is to use epoch-based reclamation
(EBR) \cite{FC11,neumann2015fast,kim2019mvrlu,NHP22,WBBFRS21a}. 
EBR manages version lists
quickly and simply by removing only the earliest versions from the end
of a version list when no active \rtxs\ need them.
\hedit{However, it cannot remove obsolete versions from the middle of a version list} 
\y{and it may therefore maintain  an unbounded number of obsolete versions.
In practice, this can lead to a blowup in space when there are long-lived \rtxs{}.}

To avoid this problem, Lu and Scott designed an MVGC
system GVM~\cite{lu2013generic} that can remove intermediate versions by
occasionally traversing all version lists and removing obsolete
versions.  The HANA system uses a similar approach~\cite{HANAgc}.
The problem with these approaches is that they scan all lists 
even if they are rarely or never updated.  B{\"o}ttcher et
al. improve on these techniques in the Steam system~\cite{bottcher2019scalable} by tying the scanning of lists to updates
on the list.  Experiments with Steam~\cite{bottcher2019scalable} 
showed that removing intermediate versions can significantly improve
space usage and throughput compared to some previous MVGC methods in
workloads with lengthy \rtxs\ and frequent updates.  However, \y{none of these}
systems provide strong worst-case time bounds and for $P$ processes
and $M$ version lists achieve bounds of at best $O(PM)$ space.


Ben-David et al.~\cite{BBF+21} recently gave a MVGC
technique (henceforth \bbf{}) that has the best known time and space bounds. 
It removes intermediate versions
without traversing the full list.
It is wait-free,
maintains only a constant factor more versions than needed, and uses
$O(1)$ steps on average per allocated version.  To achieve this,
\bbf{} defined a {\em range tracking} object to precisely identify
obsolete versions.  However, \bbf{} has not previously been
implemented and it is unclear if it can be made time-efficient
in practice.
It uses a sophisticated 
doubly-linked list
implementation, which we call \TreeDL\ because it has an implicit tree
laid over it. 

We study the efficiency of several MVGC schemes in
practice by implementing and comparing them in a single system.
Based on our experiments, we  develop two novel MVGC schemes that
combine and extend techniques from both Steam and \bbf{} to provide
protocols that are as space-efficient as \bbf{} while achieving
time efficiency comparable to Steam and EBR (sometimes slightly better, sometimes
slightly worse). 
\y{We use the MVGC schemes to unlink obsolete versions from version lists,
allowing them to be reclaimed by a garbage collector later.}
The \y{new schemes} relax some of the worst-case time
guarantees of \bbf{} to improve time performance in the common case.
However, they maintain the strong worst-case space bounds proved
in~\cite{BBF+21}.
Our experimental analysis reveals the practical
merit of state-of-the-art techniques for MVGC.  
\hedit{It
shows} that their performance  depends
significantly on the workload and that  no one technique
always dominates for both time and space.
\hedit{However, our new schemes show the most consistent performance without any particularly adverse workloads.}

Starting from the state-of-the-art MVGC schemes
in~\cite{BBF+21,bottcher2019scalable}, we had to combine,
simplify \hedit{and} 
extend  techniques to obtain the two new MVGC
schemes (as discussed in Section~\ref{sec:mvgc}). 
\y{They focus on
speeding up the removal of obsolete versions once they are identified.
This is an important part of MVGC.
Our removal algorithms greatly sped up the theoretically efficient BBF+ scheme. 
One of them was also used to come up with a lock-free version of Steam.
This version is much more efficient than the original, which locked entire version lists, for the lock-free multiversion data structures we use in our experiments.
}
\y{Our algorithms use the range tracking (\RTO) object from \bbf{} to 
identify
obsolete versions. A main challenge is that removing these versions can be 
concurrent with other operations on the version list.}

 
For the first \hedit{new} scheme, called
\RTDLGC, we develop a simpler doubly-linked list
implementation supporting removal from the middle of version lists
(see Section~\ref{sec:dll}). The implementation, called \DL\
(Practical Doubly-linked List), sacrifices \TreeDL's guarantee of
constant amortized time per operation and instead guarantees $O(c)$
amortized time, where $c$ is the number of concurrent removals of
consecutive list nodes.  
This choice was driven by
the experimental insight that removing chains of adjacent nodes
is very rare in practice: in our experiments the
average $c$ was no more than 1.01 across a wide variety of workloads. 
Experiments show the new, simpler list is much faster
than the 
constant-time list \TreeDL\ when $c$ is small.

Our second algorithm,  \RTSLGC, further improves practical
  performance for many common workloads, in which version lists tend
  to be very short, so that linear searches are reasonably efficient.
In this case, we use a singly-linked list, which requires less space (no back pointers) and less time for
updates due to fewer pointer changes.  GVM~\cite{lu2013generic},
Steam~\cite{bottcher2019scalable} and HANA~\cite{HANAgc} also rely on
this observation~\cite{bottcher2019scalable}.  We develop
a singly-linked list implementation, called
\SL\ (Simple Singly-linked List), in Section~\ref{sec:sll} and use it to implement version
lists in \RTSLGC.  
Unlike those in
Steam and HANA it is lock-free, 
and unlike GVM, which is lock-free, it does not require any mark bits
on pointers.  Such mark bits require an extra level of indirection in
\hedit{languages} that do not allow directly marking pointers.

Our new list structures support sorted
lists with appends on one end, and deletes anywhere, and are likely of
independent interest. We prove both correct.
Combining  our 
new lists with the range tracking  of \cite{BBF+21}
significantly improves space in many workloads over previous
 MVGC schemes. The range tracking object 
identifies
list elements 
to remove (in \RTDLGC) or lists  to traverse and collect (in \RTSLGC)
more accurately
than traversing a list on every update as in Steam or traversing all lists periodically
as in HANA and GVM.



In our experiments (Section~\ref{sec:experiments}), we implemented and studied several 
schemes: an epoch-based collector~\cite{WBBFRS21a}, \bbf{}~\cite{BBF+21},
our new MVGC schemes, and an optimized variant of Steam \cite{bottcher2019scalable} we
developed using \SL.
MVGC schemes encompass several mechanisms that may have a crucial effect on their time or space overhead.
We focussed on the following mechanisms that our experiments showed to have high impact on performance:
(1) data structure for storing  versions 
and (2) how to choose when to \er{remove obsolete versions.}

\y{Our experiments study these MVGC schemes in the context of multiversion concurrent data structures.
These are data structures that leverage version lists to support atomic read-only transactions 
(e.g., range queries) alongside single-key insert, delete and lookup operations.
They can be used as database indexes~\cite{wang2018building, leis2013adaptive} 
and have gained a lot of recent attention~\cite{NHP22,WBBFRS21a,SRP22,KKW22,SRP22}.}

We experimentally compared \y{the MVGC schemes} on two quite
different concurrent data structures, a \hedit{multiversion} balanced binary search tree
and a \hedit{multiversion} hash table.  As in previous work~\cite{bottcher2019scalable}, we
saw that reclaiming intermediate versions is vital for reducing memory
overhead (space), especially when there are long \rtxs\ or
oversubscription (i.e., when there are more threads than available logical cores).
In particular, EBR, which does not collect intermediate versions,
requires up to 10 times more space than the others.  Perhaps
surprisingly, and unexpected to us, Steam sometimes has particularly bad space usage
for trees even though it does collect
intermediate versions.  We found this is due to versioned
pointers pointing to objects containing other versioned pointers, as
described in Section~\ref{sec:experiments}.  This leads to cases
that require as much as 8 times more space.  We found that
\RTSLGC{} almost always performed best in terms of space.

Steam and EBR typically performed best
in terms of update throughput, but for rtx throughput the performance was
mixed.   For combined throughput there was little
difference among the schemes; sometimes Steam and EBR are
better and sometimes \RTSLGC{} and \RTDLGC{} are better.  \bbf{} is
almost always the slowest.  

In conclusion,
the experiments indicate MVGC schemes using range tracking 
avoid the high space anomalies of EBR and Steam, 
with 
throughput that is similar 
on mixed workloads and only slightly worse  on 
update-heavy workloads.

\y{Many existing practical MVGC schemes are designed for database systems that use multiversioning and are often heavily intertwined with other parts of} \er{those multiversion
systems.}
\y{We believe many of the MVGC techniques are common across both database systems and multiversion data structures.
A previous study~\cite{bottcher2019scalable} compares MVGC schemes in various database
systems (including~\cite{bottcher2019scalable,HANAgc,larson2013hekaton}) that
use different transaction management protocols.
For our experiments, we controlled for these confounding factors by applying 
different MVGC schemes to the same pair of multiversion data structures.
We know of no previous comprehensive, apples-to-apples
comparison of MVGC techniques.

}

We summarize the paper's contributions as follows.
\begin{itemize}[leftmargin=3mm,labelwidth=0mm,labelsep=1mm,topsep=0mm]
\item Two novel MVGC schemes that borrow and extend ideas from  state-of-the-art 
space- and time-efficient schemes.
The new schemes maintain the strong worst-case space bounds provided by \bbf~\cite{BBF+21},
while achieving throughput comparable to Steam~\cite{bottcher2019scalable} in most cases.
\item An experimental analysis that {\em fairly} compares 
state-of-the-art 
MVGC schemes by applying them  to the same multi\-version data structures.
This sheds light on the practical merit of the schemes, and the reasons that no
single MVGC approach is a clear winner in terms of both space and time.
\item Several insights
that can drive the design of new MVGC schemes, as they have for our new  schemes.
\item A new 
lock-free doubly-linked list \DL\ 
that is efficient when used in MVGC,
and which may be useful elsewhere. 
\item A new simple singly-linked list implementation \SL\ that is suitable for MVGC
and has better throughput than \DL. 
\end{itemize}


%% file: related.tex

\section{Background and Related Work}
\label{sec:related}

{\bf Multiversioning.} \hedit{In general, using  multiversioning involves maintaining a global timestamp representing the current time.
Depending on the implementation, this global timestamp can be incremented by \rtxs{}~\cite{WBBFRS21a,FPR19}, update operations~\cite{NHP22}, or system clocks~\cite{KKW22, lim2017cicada}.
Each update marks any new object (or version of an object) with its timestamp, and adds the new object to the head of the appropriate version list.
Each rtx reads the global timestamp and uses it to navigate the version lists.
More specifically, an rtx with timestamp $t$ traverses a version list until it finds a version whose timestamp is at most $t$, and then reads the desired value from that version.
The details of how update operations and \rtx{}s are implemented differ between implementations, but this high-level picture is sufficient for our study of MVGC.}

\hedit{A {\em versioned CAS object}~\cite{WBBFRS21a} is a specific implementation of multiversioning that we use in our experiments when comparing different MVGC schemes.
It is a CAS object that supports looking up} \er{older values previously stored in it} \hedit{given a timestamp.
Our experiments use 
these objects to add support for \rtxs\
to a CAS-based lock-free balanced binary search tree~\cite{BER14} and a
simple lock-free hash table.} \Youla{This is the first time, I think, that we use the term "historical value". 
Also, I think we do not use it in later sections. Better avoid it?}
\Hao{I changed "historical values" to "older versions"}
\Eric{Some overlap between this paragraph and parag 12 of Sec 1.  Should we mention here that allowing querying historical values also adds support for rtxs?}
\Hao{I think the last two sentences is sufficient for saying that querying historical values adds support for rtxs.}

\noindent
{\bf Epoch-based reclamation (EBR).}
EBR \cite{Fra04,B15} \hedit{is a memory reclamation technique 
	 which}
divides the execution into epochs. 
\hedit{Each operation begins by reading and announcing the current epoch via an announcement array.
When all active processes have announced the current epoch, the global epoch counter is incremented and a new epoch begins.
This epoch counter is separate from the timestamps used for multiversioning.
EBR ensures that all active operations started during either the current 
or \y{the} previous epoch, and therefore any nodes removed during earlier epochs are safe to reclaim.
This idea can be extended to work for multiversioning by observing that a \rtx{} will never access any version that was overwritten before the start of the \rtx{}.
Therefore, all versions that were overwritten before the previous epoch are safe to reclaim, as they are no longer needed by any active \rtx{}s and are also not on the path to any needed versions.
}

Because EBR is simple and fast,  variants of it
are widely used for MVGC~\cite{FC11,neumann2015fast,WBBFRS21a,Wu17}.
However, EBR-based MVGC schemes reclaim only the oldest versions from the end of the version list, and not obsolete versions in the middle of the list.
As a result, EBR does not guarantee
space bounds, and can result in high space overhead, particularly if 
processes execute long \rtxs, \hedit{thus preventing the epoch from being advanced}, 
or if processors are oversubscribed.

\noindent
{\bf Compaction-based reclamation.}
To address EBR's weakness,  
compaction-based MVGC schemes~\cite{lu2013generic,HANAgc,bottcher2019scalable} identify
obsolete versions and remove them from version lists.  A version's
\emph{time interval} is the interval of timestamps
when it was the latest version.
\hedit{
In compaction-based schemes, \rtxs{} announce the timestamp they intend to use for traversing the version lists.
A version list is compacted by first reading the announced timestamps, sorting them, and then traversing the version list to identify and remove versions whose time intervals
do not include any announced timestamp.}
Since the versions are in
order, the traversal takes constant time per version.  Steam~\cite{bottcher2019scalable}
traverses 
and compacts a version list whenever a new version is added to
it.  
HANA~\cite{HANAgc} and GVM~\cite{lu2013generic} do not tie
the \er{compaction} to adding versions, but instead either go through all
lists using background threads, or have the main threads \er{compact} once
in a while.  
\er{These techniques have the disadvantage of}
traversing a version list even if it contains few or no obsolete versions
to be collected.  
\er{HANA and GVM even visit lists when no
changes have been made to them.}
GVM~\cite{lu2013generic} and
Steam~\cite{bottcher2019scalable} guarantee that 
each version list contains $O(P)$ versions, where $P$ is the number of processes, so \hedit{$O(PM)$ versions are maintained for $M$ lists}.
%
%


\noindent
{\bf Range-tracking.}
\bbf{}~\cite{BBF+21} uses a range tracking object to 
directly identify obsolete versions, 
avoiding the traversal of the entire list. 
A range tracking object tracks a set of non-current versions, each of which
has been assigned an integer range indicating its time interval.
A version in the range
tracking object can be reclaimed if its interval does not intersect any of
the \hedit{timestamps announced by \rtxs{}.} 
In \RT\ (the linearizable range tracking implementation of~\cite{BBF+21}),  
each thread $p$ appends non-current versions in a local list. 
If the size of $p$'s list becomes $P\log P$, $p$ performs a {\em flush}.
The flush appends
$p$'s local list to a shared FIFO queue $Q$ (of lists)~\cite{FK14}.  It then dequeues two lists from $Q$,
merges their contents and compares the merged list against a sorted sequence of the current announcements to determine which versions
 are still needed.  The list of needed versions is enqueued to $Q$, and the set of  obsolete versions is returned so that they can be removed from their lists.
A flush phase requires $O(P\log P)$ steps and is performed once
every $\Theta(P\log P)$ times an element is inserted in the local list. 
\RT\ thus \hedit{ensures} amortized constant time for each operation it supports.

\noindent
{\bf Removing from Lists.}
After identifying an obsolete version,
it must be spliced
out of its list.  Steam~\cite{bottcher2019scalable} and HANA~\cite{HANAgc} use locks to 
safely do so. 
GVM~\cite{lu2013generic} uses a variant of Harris's lock-free singly-linked list~\cite{harris2001pragmatic}. 
This works since lists are always traversed from their head.  However, 
\bbf{} requires
safely removing obsolete nodes from the middle of a list, given only
a pointer to the node.  To do this safely and efficiently, \bbf{}
uses \TreeDL, a custom wait-free doubly-linked list. 
\TreeDL\ uses an implicit binary  tree
whose in-order traversal gives the list nodes in order. List nodes that are
leaves of the tree are not adjacent in the list and can
therefore be removed safely \hedit{and concurrently}.  
If only leaves are removed, then
leaves that store versions that are still needed by \rtxs\ may 
prevent obsolete versions at internal nodes from being removed, resulting in high space bounds. 
Thus, \TreeDL\  provides intricate helping mechanisms
that also permit the removal of internal list nodes 
in the tree.
This is done with care 
to ensure list consistency and \hedit{low space usage}.
%


\noindent
{\bf Other reclamation schemes.}
\hedit{A recent improvement on EBR is version-based reclamation~\cite{SHP21}. It bounds memory by restarting long running operations, and has very recently been applied to a multiversion data structure~\cite{SRP22}.}
Other memory reclamation schemes have been
proposed, 
but  they do
not solve the MVGC problem~\cite{HL+05,M04,RC17,WI+18-I}, or 
require special support, in hardware~\cite{AE+14,DH+11} or through the
operating system~\cite{B15,SBM21}.  \bbf{} uses Hazard
pointers~\cite{M04} and a recent implementation~\cite{ABW21} of
reference counting~\cite{correia2021orcgc,DM+01,HL+05} to 
deallocate nodes that have been spliced out of the 
lists.
Our new MVGC schemes rely on the automatic Java garbage
collector to deallocate unreachable nodes.

\noindent
{\bf Other lock-free list structures.}
We designed \SL\ and \DL\ as special-purpose lock-free linked lists for representing version lists.
They avoid the extra level of indirection of Valois's singly-linked list \cite{Val95}, which
places auxiliary nodes between real nodes.
Harris's singly-linked list \cite{harris2002practical} and others that build on it \cite{FR04,Mic02} also add a level of indirection when
implemented in Java, since they require a mark bit on list pointers.
Besides \TreeDL, discussed above,
other existing lock-free doubly-linked list implementations are quite complex 
because they support more operations than we need for  version lists \cite{ST08,Sha15}
or rely on multi-word CAS instructions \cite{AH13,Gre02}, which are not widely available in hardware
(although they can be simulated in software~\cite{harris2002practical,guerraoui2020efficient}).

%% file: preliminaries.tex

\section{Proposed MVGC \y{Schemes}}
\label{sec:mvgc}

We propose two new MVGC \y{schemes}, \RTDLGC\ and \RTSLGC.
\y{As with BBF+, we use the \RT{} to identify which versions in the middle of a list can be removed. }
\hedit{Our new schemes can be applied in the same way as BBF+; namely, \rtxs{} announce and unannounce their timestamps using the operations provided by the \RT{}, and whenever a version} \y{is} \hedit{overwritten by a newer version, it is passed to the \RT{} along with its time interval. \Youla{Possibly better talk about operations of \RT{} rather than functions of \RT{}?}
\Hao{I changed functions -> operations}
The two new schemes preserve the correctness conditions (e.g., linearizability, sequential consistency) of the multiversion data structure they are applied to.
}

\hedit{Our first scheme,} \RTDLGC\ is based on a novel doubly-linked list implementation, called \DL,
whereas \RTSLGC\ employs a new, simple implementation of a singly-linked list, called \SL. 
The lists are sorted by a key field (e.g., the version timestamp).
\DL\ supports the following operations, introduced in \cite{BBF+21,WBBFRS21a}:
1) \var{tryAppend(x,y)} attempts to append a new element \var{y} 
at the end of the list, given that \var{x} is currently the last list element; 
2) \var{remove(x)} removes the element \var{x} from the list (and is used for garbage collection); 
3) \var{peekHead} returns the last element of the list (i.e., the current version);
4) \var{search(key)} returns the latest element of the list whose key is less than or equal to \var{key}
(i.e., the version that was current when  \var{key} was the current timestamp).
Consistent with how MVGC uses version lists, 
\DL\ assumes a \var{remove} is called only once on each node,
and not on the head node.

%
The list interface above supports {\em versioned CAS objects}, as described in
\cite{WBBFRS21a}.  In this case, a CAS operation $op$ that wants to change the value from
\var{v} to $\var{v}'$ calls \var{peekHead}, and checks if the
value stored in the latest version is \var{v}. 
If so, $op$ calls \var{tryAppend} to add a version containing the value $\var{v}'$ after \var{v}.
If \var{tryAppend} fails, a concurrent CAS must have successfully changed the
value of the CAS object to something different from \var{v}. Then, $op$
returns false.
%

\er{To avoid some of the intricacies of \TreeDL\,
we} designed our novel implementation \DL\ based on a simple idea:
to remove a node we  mark it  and then traverse outwards to the first unmarked
node in either direction and make them point to each other.
\DL\ ensures that at most  
$L-R + P$ nodes in total are in the version lists at the end of an execution 
containing a total of $L$ appends and $R$ removes. 
This bound is better than the $2(L-R) + O(P\log L_{max})$ bound provided by \TreeDL~\cite{BBF+21}, \hedit{where $L_{max}$ is the maximum number of appends on a single version list}.
In terms of time, \DL\ ensures that \tryAppend\ and \var{peekHead}
maintain the $O(1)$ bound of \TreeDL, but the   number of steps for \remove\ is $O(c)$,
where $c$ is the number of concurrent removals of consecutive nodes in the list.
The average value of $c$ was at most 1.01 in each of our experiments.
Thus, from a practical perspective, relaxing the time bound for \remove\
to $O(c)$ is a good compromise. 

\SL\ is a singly-linked list implementation, and thus lacks back pointers.
In a doubly-linked list, a version can be removed 
by accessing only the node's immediate neighbourhood.
To remove a node from a singly-linked list, we must find its predecessor in the list
by traversing the list from the head.
Thus, instead of focusing on removing individual nodes, \SL\ supports a \compact\ routine
that traverses the entire list, removing obsolete nodes.
Experiments showed that, in many cases, the version lists tend to be short, 
so traversing a list is often inexpensive.
Indeed, in most experiments, \RTSLGC, which employs \SL, exhibits better throughput than \RTDLGC.
\SL\ ensures the same $L-R+P$ bound on the number of nodes contained in the version lists as \DL.

\DL\ and \SL\ 
are presented in Sections~\ref{sec:dll} and~\ref{sec:sll}.
Due to lack of space, their correctness proofs and complexity bounds are in the \er{\arxcam{appendices}{full version \cite{arxiv}}}. 
%
We used \DL\ to develop \RTDLGC, and \SL\ to develop \RTSLGC,
as well as \STEAMLF, a lock-free version of Steam that we describe below.
Specifically, instead of using \TreeDL\ for implementing  the version lists as in \bbf, 
\RTDLGC\ employs \DL\ and \RTSLGC\ employs \SL.
\RT~\cite{BBF+21} is used by both
\RTDLGC\ (to decide when to splice out an individual node), and by
\RTSLGC\ (to decide when to traverse and compact a version list).  

\RT\ is proved in~\cite{BBF+21} to use $O(H + P^2\log P)$ space, where $H$ is the \hedit{maximum} number of needed versions \hedit{at any point during the execution}.
We call the left and right pointers of nodes in \DL\ and the left pointers of nodes in \SL, {\em access pointers}. 
A node is {\em reachable} in \DL\ or in \SL\ at some point in time, if it can be reached
starting from the list head and following access pointers. 
Recall that both \DL\ and \SL\ ensure that at most  
$L-R + P$ nodes in total remain reachable in the version lists. 
This bound and the bound for \RT\ imply the following. 

\begin{theorem}
\label{thm:mvgc-space}
Consider an implementation of a concurrent multi-versioning data structure that uses \RTDLGC\ or \RTSLGC\ for garbage collection. 
At any time $t$ of an execution, the total number of versions that are reachable in the version lists is $O(H + P^2 \log P$),
where $H$ is the maximum, over all times before $t$, of the number of needed~versions.
\end{theorem}

\hedit{
If the number of needed versions is large at} \y{some} \hedit{point in} \y{an} \hedit{ execution}, 
\y{this} \hedit{will only influence the size of the version lists for a limited number of steps.
More specifically, if the number of needed versions remains below some quantity $h$ in a suffix of an execution, 
then the number of reachable versions in Theorem~\ref{thm:mvgc-space} will eventually be $O(h + P^2\log P)$.}
\Eric{I reworded slightly}
\Youla{Also, it is not clear why the claims here are true.}
\Hao{we could add "this can be shown by slightly modifying the space bound proof in ~\cite{BBF+21}", but I'm not sure if that is more satisfactory. I don't think we have enough space here to explain why the claim is true.}
%
%


\ignore{
\begin{theorem}
\label{thm:mvgc-time}
Consider an implementation of a concurrent multi-versioning data structure that uses \RTDLGC\ for garbage collection. 
1) The amortized time complexity of \op{tryAppend}, \op{getHead}, and creating a new version list is $O(1)$.
2) The amortized time complexity of \op{remove} is $O(c)$.  
\end{theorem}
}

%% file: doubly-linked-list.tex

\section{Doubly-Linked List}
\label{sec:dll}

\ignore{
\Eric{This parag may now overlap with earlier material.  Should mention the acronym \DL\ somewhere in opening}
Garbage collection requires removing a version identified as obsolete from its version list.
If the list is short, it may be acceptable to traverse a singly-linked list to remove
a version.
However, if the list is long, a doubly-linked list will allow more efficient removals.
This approach was used in \cite{BBF+21} to provide version lists  where the 
average number of steps for a \remove\ operation was $O(1)$.
Achieving this bound required a complex algorithm.
In Algorithm \ref{alg:doubly-linked-list} we give a  simpler doubly-linked version list 
that has a weaker time bound, but achieves good performance in our experiments.
}

Algorithm \ref{alg:doubly-linked-list} provides the pseudocode for \DL. 
Each node stores a \var{left} and \var{right} pointer, as well as 
a \var{mark} bit to facilitate deletions. It also stores a key \var{key}
and a value $\var{val}$, which are set when the node is created and never change.
List elements are appended to the right end of the list and  a \var{head}
pointer points to the rightmost node.
The list  initially contains a sentinel node with key $-\infty$, which remains at the left end of the list at all times.

\newcommand{\tallstrut}{\rule{0pt}{1.05\baselineskip}}
\renewcommand{\figurename}{Algorithm}

\begin{figure}
  \begin{lstlisting}[linewidth=.99\columnwidth, numbers=left]
class Node {
  int key; Value val; bool mark; // initially false
  Node* left, right; }  <-- OMIT CONSTRUCTOR TO SAVE SPACE
  Node(Key k, Value v) { // constructor
    val = v; key = k; mark = false;
    left = null; right = null; } } -->
class DoublyLinkedList {															@\tallstrut@
  Node* head; <-- OMIT CONSTRUCTOR TO SAVE SPACE
  DoublyLinkedList() { // constructor
    head = new Node(-inf, $\bot$); } // sentinel node
-->
  Value peekHead() { return head->val; }				@\label{peekHead-return}@ 	@\tallstrut@
  Value @\search@(int k) { 															@\tallstrut@
    Node* x = head;										@\label{search-begin}@
    while(x->key > k) {									@\label{search-test}@
      x = x->left; }									@\label{search-advance}@
    return x->val; } }
  bool @\tryAppend@(Node* x, Node* y) {												@\tallstrut@
  	Node* w = x->left;							@\label{append-read-left}@
  	// first, help tryAppend(w, x) if necessary
  	if(w != null) CAS(&(w->right), null, x); 	@\label{help-right}@
  	y->left = x; 									@\label{init-left}@
  	if(CAS(&head, x, y)) { 							@\label{CAS-head}@
  		CAS(&(x->right), null, y);					@\label{set-right}@
  		return true;
  	} else return false; }
  void @\remove@(Node* x) {															@\tallstrut@
    x->mark = true;											@\label{mark}@
    Node* left = x->left;									@\label{first-left}@
    Node* right = x->right;									@\label{first-right}@
    Node* leftRight, rightLeft;
    while(true) {
      while(left->marked) left = left->left;				@\label{advance-left}@
      while(right->marked) right = right->right;			@\label{advance-right}@
      rightLeft = right->left;								@\label{right-left}@
      leftRight = left->right;								@\label{left-right}@
      if(left->marked || right->marked) continue;			@\label{remove-test}@
      if(!CAS(&right->left, rightLeft, left)) continue; 	@\label{CAS-left-DLL}@
      if(!CAS(&left->right, leftRight, right)) continue;	@\label{CAS-right}@
      break; }  }
  \end{lstlisting}
\caption{Doubly linked list implementation (\DL{}).}
\label{alg:doubly-linked-list}
\end{figure}

\renewcommand{\figurename}{Figure}


The \var{peekHead} operation simply reads \var{head} and returns the node's \var{val} field. 
A \var{search(k)}  starts at \var{head} and follows \var{left} pointers until reaching a node whose key is at most \var{k}.

A \var{\tryAppend(x,y)} attempts to append a new node \vy\ after
a node \vx\ that has previously been read from the \var{head} pointer.
It updates the \var{left} pointer of \vy\ to \vx\ (line \ref{init-left}) and swings the \var{head}
pointer from \vx\ to \vy\ using a CAS  (line \ref{CAS-head}).
If swinging the \var{head} fails, the operation returns false.
Otherwise, \tryAppend\ attempts to update the \var{right} pointer of \vx\ to \vy\
using the CAS at line~\ref{set-right}.
If swinging the \var{head} succeeds, but \tryAppend\  pauses before updating \vx's \var{right} pointer,
the list is left in an inconsistent state.
So, before any subsequent \tryAppend\ adds a node beyond \vy, it first
helps complete the append of node \vy\ by updating \vx's \var{right} pointer (line \ref{help-right}).
The use of CAS (lines~\ref{help-right} and~\ref{set-right}) 
to initialize \vx's \var{right} pointer 
ensures that the pointer is initialized only once.

A \var{\remove(\vx)}  first marks the node \vx.
(Line~\ref{mark} does this with a simple write, since a marked node remains marked forever.) 
It then seeks the nearest unmarked nodes on either side and updates them
to point to each other using CAS instructions.
If either CAS fails (or if the validation at line \ref{remove-test} 
that the two nodes are still unmarked fails)
the \remove\ operation tries again to find the nearest unmarked nodes, continuing
outwards from where the previous attempt left off.
Once both pointers are updated, the \remove\ terminates.

\begin{figure}
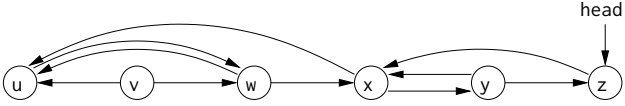
\caption{An example state of the list.  A \var{remove(v)} has completed.  A pending \var{remove(\vw)} and \var{\remove(\vy)} have performed the CAS at line \ref{CAS-left-DLL} successfully but not the CAS at line \ref{CAS-right}. A \search\ may be visiting any of the nodes shown.\label{fig-inconsistent}}
\end{figure}

\DL\ does not maintain perfect consistency between \var{left} and \var{right} pointers:
a \remove\  may update one but pause or  die before updating the other.
Moreover, to keep the implementation lightweight, 
a \remove\  does not help  other nearby \remove\ operations to make the list consistent before applying
its own changes; it simply makes its changes (which may, incidentally, help other \remove\ operations take effect).
See Figure \ref{fig-inconsistent} for an example.
In addition, a \search\ concurrent with \remove\ operations may traverse removed nodes.
Thus, a careful proof is needed to show some fairly weak invariants
suffice for linearizability~\cite{HW90}.
We sketch the proof in Section~\ref{DLL-correctness}; for details see \er{\arxcam{Appendix~\ref{DLL-appendix}}{the full version \cite{arxiv}}}. 

\subsection{Correctness}
\label{DLL-correctness}

We show that any execution $\alpha$ of the doubly-linked list is linearizable.
We assume the following preconditions.
When \var{\tryAppend(x,y)} is called, \vy\ is a new node that has never been used as the second argument of \tryAppend\ before and it is not the sentinel node.  Moreover, \vx\ has been read from \var{head} and \vy's key is greater than or equal to \vx's key.
There is at most one call to \var{remove(x)} for each \vx, and it can be called only after a call to \var{\tryAppend(x,*)} has returned true.

If \vx\ and \vy\ are pointers to nodes,
we use the notation $\vx \leftarrow \vy$ to indicate that \var{\vy->left = \vx} and 
$\vx\rightarrow \vy$ to indicate that \var{\vx->right = \vy}.
We use $\vx \lreach \vy$ to indicate there is a path of \var{left} pointers from \vy\ to \vx.
We say that \vx\ is added to the list when the \var{head} pointer is set to \vx.
Since there is at most one call to tryAppend with \vx\ as the second argument, there is at most one
step in $\alpha$ that sets \var{head} to \vx.
We  define a total order on nodes that are added to the list during $\alpha$:
$\vx \prec \vy$ if the \var{head} pointer was set to \vx\ before it was set to \vy\ in $\alpha$.

The following invariant captures some relationships between a node and its neighbours.
Essentially, it says that \var{left} pointers always point to older nodes and \var{right} pointers point to newer nodes.
Moreover, if one of those pointers skips over a node, the skipped node must be marked for deletion.
To prove these invariants, we also need similar claims for the local variables
\var{left} and \var{right} of pending \var{removes}.

\begin{restatable}{invariant}{resInvOrder}
\label{inv-order}
Let $C$ be a configuration and \vy\ be a node.
\begin{enumerate}[leftmargin=2mm,labelwidth=0mm,labelsep=1mm,topsep=0mm]
\item
\label{left-order}
If $\vy$ has been added to the list (and is not the sentinel), then \var{\vy->left} is also a node that was added to the list  and $\var{\vy->left}$ $\prec \vy$.   
For all \vw, if $\var{\vy->left} \prec \vw \prec \vy$ then \vw\ is~marked.
\item
\label{right-order}
If \var{\vy->right} was non-null in a configuration at or before $C$, then in $C$ \var{\vy->right} is a node that was added to the list
and $\var{\vy->right} \succ \vy$.
For all \vw, if $\vy \prec \vw \prec \var{\vy->right}$ then \vw\ is marked.
\item
\label{remove-order}
For any call to \var{\remove(\vy)} in progress, the local variables \var{left} and \var{right} are nodes that have been added to the list and $\var{left} \prec \vy \prec \var{right}$.
Moreover, for all $\vw$, if  $\var{left} \prec \vw \prec \var{right}$ then \vw\ is marked.
\item
\label{sentinel-left}
The \var{left} pointer of the sentinel node is null.
\end{enumerate}
\end{restatable}

The \var{head} pointer is initialized to the sentinel node,
and can only be updated on line \ref{CAS-head} to a non-null pointer.
This fact, Invariant \ref{inv-order}
and the preconditions of the operations imply that a null pointer is never dereferenced.

Because \remove\ operations update \var{left} pointers concurrently, and may splice out multiple nodes at once, 
there is a danger that one \remove\ may undo the effects of another. For example, suppose we have
four nodes $\vw\leftarrow \vx\leftarrow \vy \leftarrow \vz$.  If a \var{remove(\vx)} updates \var{\vy->left}
to \vw\ and a concurrent \var{remove(\vy)} updates \var{\vz->left} to \vx, then \vx\ would remain reachable.
The following lemma ensures that this cannot happen.  (In fact, it is more general in that it applies even when the four nodes are not consecutive in the list.)

\begin{restatable}{lemma}{resNoCrossovers}
\label{no-crossovers}
Suppose $\vw\prec\vx\prec\vy\prec\vz$ and $\vw \leftarrow \vy$ at some time during the execution.
Then there is never a time when $\vx \leftarrow \vz$.
\end{restatable}

Invariant \ref{inv-order} can be used to show 
that each time a \var{left} pointer changes, it points to an older node.

\begin{restatable}{lemma}{resCASAdvances}
\label{CAS-advances}
If line \ref{CAS-left-DLL} does a \var{CAS(\vz->left, \vy, \vw)}, then
$\vw \preceq \vy$.
\end{restatable}

The \emph{abstract list} $AL$ in a configuration 
is the sequence of nodes reachable from the head by following \var{left} pointers.
The next three lemmas ensure that a node \vx\ is  removed from $AL$ once 
\var{\remove(x)} has terminated, and only if
\var{\remove(x)} has been invoked.
Lemma \ref{remove-only} implies that updates to a \var{left} pointer only remove elements from $AL$:  the list after the update is a subsequence of the list prior to the update.

\begin{restatable}{lemma}{resRemoveOnly}
\label{remove-only}
If a CAS at line \ref{CAS-left-DLL} changes \var{\vz->left} from \vy\ to \vw, then
$\vw \lreach \vy$ in the preceding configuration~$C$.
\end{restatable}

\begin{restatable}{lemma}{resRemoveSafe}
\label{remove-safe}
When a node is removed from $AL$, it is~marked.
\end{restatable}

\begin{restatable}{lemma}{resRemoveSucceeds}
\label{remove-succeeds}
When \var{\remove(x)}  terminates, \vx\ is not in $AL$.
\end{restatable}

These lemmas can be used to verify that the following linearization is correct.
We linearize a successful \tryAppend\ when it updates the \var{head} pointer.
We linearize a \var{\remove(x)} when a CAS at line \ref{CAS-left-DLL} causes
\vx\ to be removed from $AL$.
If several nodes are removed by a single CAS, we order them in decreasing order by $\prec$.
Thus, at all times, the abstract list $AL$ agrees with the update operations linearized so far.
The following lemma, which requires a technical proof, allows us to linearize \search\ operations.

\begin{restatable}{lemma}{resSearchInv}
\label{search-inv}
Let $C$ be a configuration and \vx\ be the local variable of a pending \var{\search(k)} in $C$.
There was a time between the invocation of the \search\ and $C$ when \vx\ was in $AL$ and either \vx\ was the first node in $AL$ or its predecessor in $AL$ had a
key greater than \var{k}.
\end{restatable}

Thus, we can linearize \var{\search(k)} when the returned node \vx\ is the first  in $AL$ with key at most~\var{k}.

A \tryAppend\ takes $O(1)$ steps.  A \search\ is wait-free, and \remove\ is lock-free.
A \remove(\vx)\ operation takes $O(c)$ steps, where $c$ is the length of the chain of adjacent
nodes containing \vx\ that are marked before the \remove\ ends.  Here,  \vy\ and \vz\ are \emph{adjacent} if, during
the \remove(\vx),  $\vy\leftarrow \vz$ or $\vy \rightarrow \vz$.

%% file: inconsistent-dll.pdf_t
\begin{picture}(0,0)%
\includegraphics{inconsistent-dll.pdf}%
\end{picture}%
\setlength{\unitlength}{2763sp}%
\begingroup\makeatletter\ifx\SetFigFont\undefined%
\gdef\SetFigFont#1#2#3#4#5{%
  \reset@font\fontsize{#1}{#2pt}%
  \fontfamily{#3}\fontseries{#4}\fontshape{#5}%
  \selectfont}%
\fi\endgroup%
\begin{picture}(5566,928)(1043,-1718)
\put(6226,-961){\makebox(0,0)[lb]{\smash{{\SetFigFont{8}{9.6}{\rmdefault}{\mddefault}{\updefault}{\color[rgb]{0,0,0}\var{head}}%
}}}}
\put(1126,-1636){\makebox(0,0)[lb]{\smash{{\SetFigFont{8}{9.6}{\rmdefault}{\mddefault}{\updefault}{\color[rgb]{0,0,0}\var{u}}%
}}}}
\put(3226,-1636){\makebox(0,0)[lb]{\smash{{\SetFigFont{8}{9.6}{\rmdefault}{\mddefault}{\updefault}{\color[rgb]{0,0,0}\vw}%
}}}}
\put(2176,-1636){\makebox(0,0)[lb]{\smash{{\SetFigFont{8}{9.6}{\rmdefault}{\mddefault}{\updefault}{\color[rgb]{0,0,0}\var{v}}%
}}}}
\put(4276,-1636){\makebox(0,0)[lb]{\smash{{\SetFigFont{8}{9.6}{\rmdefault}{\mddefault}{\updefault}{\color[rgb]{0,0,0}\vx}%
}}}}
\put(5326,-1636){\makebox(0,0)[lb]{\smash{{\SetFigFont{8}{9.6}{\rmdefault}{\mddefault}{\updefault}{\color[rgb]{0,0,0}\vy}%
}}}}
\put(6376,-1636){\makebox(0,0)[lb]{\smash{{\SetFigFont{8}{9.6}{\rmdefault}{\mddefault}{\updefault}{\color[rgb]{0,0,0}\vz}%
}}}}
\end{picture}%

%% file: singly-linked-list.tex
 
\section{Singly-Linked List Compaction}
\label{sec:sll}

\newcommand{\vt}{\var{t}}
\newcommand{\vA}{\var{A}}
\newcommand{\ts}{\var{ts}}
\newcommand{\vh}{\var{h}}
\newcommand{\vi}{\var{i}}
\newcommand{\vj}{\var{j}}
\newcommand{\vk}{\var{k}}
\newcommand{\vu}{\var{u}}
\newcommand{\Announce}{\var{Announce}}
\newcommand{\head}{\var{head}}
\newcommand{\AnnScan}{AnnScan}
\newcommand{\GlobalAnnScan}{\var{GlobalAnnScan}}

\renewcommand{\figurename}{Algorithm}

\begin{figure}
	\begin{lstlisting}[linewidth=.99\columnwidth, numbers=left]
class Node { Node* left; int ts; Value val;}
class AnnScan { List<int> A; int t; }
AnnScan* scanAnnounce() {													@\tallstrut@
	AnnScan* newScan = new AnnScan;
	repeat twice {
		AnnScan* old = GlobalAnnScan;									@\label{read-globalannscan}@
		newScan.t = GlobalTimeStamp;
		read Announce[1..P] one by one, and sort into newScan.A;
		if CAS(&GlobalAnnScan, oldScan, newScan) return newScan; }		@\label{CAS-globalannscan}@
	return GlobalAnnScan }
void compact(List<int> A, int t, Node* h) {	// A is sorted	@\tallstrut@
	A.appendFront(-1); // padding
	int i = index of last element of A;											@\label{init-i}@
	Node* cur = h;													@\label{init-cur}@
	while(cur != sentinel) {											@\label{test-sentinel}@
		Node* next = cur->left;											@\label{init-next}@
		// skip nodes whose timestamps exceed threshold t
		if(cur->ts > t) {											@\label{cur-too-new}@
			cur = next;													@\label{advance-cur-1}@
		} else {
			while(A[i] >= cur->ts) i--;							@\label{advance-i}@
			if(A[i] >= next->ts) { // next is needed			@\label{test-next-needed}@
				cur = next;												@\label{advance-cur-2}@
			} else { // next is not needed
				Node* newNext = next->left;								@\label{init-new-next}@
				while(A[i] < newNext->ts) {  					@\label{test-new-next-obs}@
					newNext = newNext->left; }							@\label{advance-new-next}@
				while(!CAS(&(cur->left), next, newNext)) {				@\label{CAS-left-SLL}@
					next = cur->left;									@\label{reread-left}@
					if(next->ts <= newNext->ts) break; }					@\label{test-overshoot}@
				cur = cur->left; } } } }									@\label{advance-cur-3}@
bool tryAppend(Node* x, Node* y) {							@\tallstrut@
	y->left = x;
	return CAS(&head, x, y); }											@\label{CAS-head-SLL}@
Value @\search@(int k) {									@\tallstrut@
	Node* x = head;				@\label{init-x}@
	while(x->ts > k) {			@\label{test-x}@		
		x = x->left; }			@\label{advance-x}@	
	return x->val; }
	\end{lstlisting}
	\caption{Singly linked list (\SL{}) compaction routine.\here{Should make code a little more consistent in style to Algorithm \ref{alg:doubly-linked-list}--have a singly-linked list class etc.}
}
	\label{alg:singly-linked-list-compact}
\end{figure}

\renewcommand{\figurename}{Figure}

We now describe our simple singly-linked list \SL\ designed to store version lists
(Algorithm \ref{alg:singly-linked-list-compact}).
Each list node stores a version and has an associated timestamp \ts\ and
a pointer \var{left}
to a preceding version. 
A \var{head} pointer stores the most recently
appended node.
Like \DL, \SL\ provides a \tryAppend\ and \search\ operation.
Nodes are appended to the list with timestamps in non-decreasing order, so that the list is always sorted by timestamp.
Instead of a \remove\ operation that splices out individual nodes, \SL\ provides a \compact\ operation that
traverses the whole list splicing out obsolete nodes.
  We assume the list initially contains a sentinel
node with timestamp $-\infty$, which always remains at the list's tail.

The \compact\ operation is given a sorted list \vA\ of timestamps announced by \rtxs,
a threshold timestamp \vt\ that was read from the global timestamp counter
and a node \vh\ read from the \var{head} pointer of a version list, which specifies
where to begin traversing the list to compact.
We describe below exactly how these arguments are obtained so that
any \rtx\ that is concurrent with the \var{compact} or begins after the \var{compact}
will have a timestamp that is either in \vA\ (because the \rtx's announcement is in \vA)
or greater than or equal to \vt\ (because the \rtx\ grabs its timestamp after \vt\ has been read from the global counter).
The following definition describes the nodes that \var{compact} should retain.
A  node $\vx$ is \emph{needed w.r.t. to $\vA$ and $\vt$} (abbreviated $needed(\vA,\vt)$)~if (1) $\vx.\ts>\vt$, or (2) $\vx$ is the last appended node with timestamp at most $\vt$, or (3) for some list element $\vA[i]$, $\vx$ is the last appended node whose timestamp is less than or equal to $\vA[i]$.


Since both the version list and \vA\ are sorted by timestamp, \var{compact} makes a pass across
both the version list and \var{A} using an algorithm similar to merging two sorted arrays
to determine which nodes of the list can be removed.  When a removable node is found,
it is spliced out of the list by a CAS instruction on its predecessor's \var{next} field.
If multiple consecutive nodes are all found to be removable, a single CAS attempts to splice them all
out of the list at once.

In more detail,
the variables \vi\ and \var{cur} of \compact\ are used as pointers into \var{A} and the version list, respectively,
starting at the highest reserved timestamp and the most recent version.
Lines \ref{cur-too-new}--\ref{advance-cur-1} skip past nodes in the version list whose timestamps
are greater than \vt.
Line \ref{advance-i} finds the appropriate element of \var{A} to compare to
the timestamp of \var{cur}.
Lines \ref{test-next-needed}--\ref{advance-cur-2} skip past a node if it is needed
for the timestamp \var{A[i]}.
Lines \ref{init-new-next}--\ref{test-overshoot} handle the removal of \var{next} if it is found
to be unnecessary.
The \compact\ routine first finds the first needed node after \var{next} and stores a pointer
to it in \var{newNext} (lines \ref{init-new-next}--\ref{advance-new-next}).
It then tries to splice out a block of consecutive nodes, including \var{next}, by
changing \var{cur->left} from \var{next} to \var{newNext} at line \ref{CAS-left-SLL}.
If the CAS fails, it repeatedly tries to update 
\var{cur->left} to \var{newNext} until it succeeds or finds that \var{cur->left} points
to a node whose timestamp is less than or equal to \var{newNext}'s.
Once the block of nodes is removed, the \compact\ routine proceeds down the list (line \ref{advance-cur-3}).

We now discuss how to obtain the arguments for a \var{compact}.
\er{Each process performing a \rtx\ writes its timestamp into a global array \var{Announce}.}
If a process  simply copies the global timestamp into~\vt\ and
copies \Announce\ element by element into a local copy \vA\
before calling \compact(\vA,\vt,\vh), updates to the  list may behave strangely.
Because the copies are not made atomically, two concurrent \compact\ operations may 
disagree about the set of needed nodes:  for example, if there are 4 nodes $\vx_1,\vx_2,\vx_3,\vx_4$ in the list with timestamps
$1,2,3,4$, one \compact\ may think $\vx_1,\vx_3,\vx_4$ are needed while the other thinks $\vx_1,\vx_2,\vx_4$ are needed.  The two operations may get poised to splice out $\vx_2$ and $\vx_3$, respectively.  If they then perform their CAS steps to splice out the nodes,
node 2 will be removed by the CAS of the first operation and then become reachable again when the
CAS of the second operation occurs.
We wish to avoid this situation to maintain good worst-case space bounds.
A similar problem may occur if the copy \vh\ of the \var{head} pointer used as the starting point of a \compact\ is not obtained at the same time as \vA\ and \vt.

This problem could be avoided by taking an atomic snapshot of the global timestamp,
\var{Announce} and \var{head}, but that is too expensive.
It suffices that the intervals of time that different processes use to copy the global timestamp and  \Announce\  
are non-overlapping. This can be guaranteed using a more lightweight synchronization mechanism.
(See the \var{scanAnnounce} routine in Algorithm \ref{alg:singly-linked-list-compact}.)
Whenever a process $p$ reads the global timestamp and \var{Announce} into local copies \vt\ and \vA, it attempts to install an \AnnScan\ object, which stores the pair $(\vA,\vt)$, into
a shared variable \GlobalAnnScan\ using a CAS.  
If $p$'s CAS fails, $p$ tries again to make a local copy and install it in \GlobalAnnScan.
If $p$'s second CAS on \GlobalAnnScan\ also fails, 
\er{some other process must have stored an \AnnScan\ object there that it obtained after $p$'s first failed CAS, so
$p$} can simply use the \AnnScan\ object it finds
in \GlobalAnnScan, since it is guaranteed to be recent.
Before calling \var{\compact(\vA,\vt,\vh)}, a process gets a snapshot of just two
variables, \GlobalAnnScan\ (for the values of \vA\ and \vt) and \var{head} (for the value of \vh).


The lock-free \compact\ routine allows multiple processes to splice nodes out of the list
concurrently,  while other processes executing \var{search} traverse the list.
A careful proof of linearizability is in \er{\arxcam{Appendix \ref{SLL-appendix}}{the full version \cite{arxiv}}}. 
We linearize \var{tryAppend} operations when they update \var{head} and \var{search} operations when they read \var{head}.
We also show that \compact\ and \search\ are wait-free and that after a \var{compact(A,t,h)} routine terminates, all nodes reachable from the head of the list that were appended to the list before \vh\ are $needed(\vA,\vt)$.

%% file: experiments.tex
\section{Experiments}
\label{sec:experiments}


We experimentally compared  throughput and memory usage of
state-of-the-art MVGC schemes and our new schemes.  
We tested \RTDLGC, \RTSLGC, \bbf,
\STEAMLF, and an epoch-based scheme \EBR.  We implemented all 
for two different multiversion data structures and used the benchmarking suite
from~\cite{WBBFRS21a}.  Only the MVGC code varies.
Our code is publicly available on GitHub
\footnote{https://github.com/cmuparlay/ppopp23-mvgc}.
%

\subsection{Setup}
\noindent 
{\bf Machine/Compiler.}
Our experiments ran on a 64-core Amazon Web Service c6i-metal instance with
2x Intel(R) Xeon(R) Platinum 8375C (32 cores, 2.9GHz and 108MB L3 cache),
and 256GB memory. 
Each core is 2-way hyperthreaded, giving 128 hyperthreads. 
The machine runs Ubuntu 22.04.1 LTS.
Our experiments were written in Java and we used OpenJDK 19.0.1 \hedit{with heap size set to 64GB.
We found that overall, the runtime overhead of Java's garbage collector is small, usually accounting for less than 5\% of the overall time in most workloads.
We also tried running on smaller heap sizes down to 10GB and did not see much difference in performance.}
For each data point in our graphs, we ran the processes for 25 seconds to warm up the JVM and then measured
5 runs, each of 5 seconds.  We report the average of those 5 runs.
Error bars indicate variance.
\hedit{To measure memory usage, at the end of each run, we measure the amount of reachable memory used by the multiversion data structure. This includes metadata maintained by the data structure’s MVGC scheme.}


\noindent 
{\bf Data Structures.} 
We apply MVGC schemes to CAS-based implementations 
of two concurrent data structures by replacing each CAS object with a
versioned CAS object, as in~\cite{WBBFRS21a}.   This adds 
support for linearizable \rtxs\ \hedit{on top of the usual insert, delete, and lookup operations}. 
%
Our experiments test a hash table and a chromatic
tree~\cite{BER14}, which
have quite different characteristics, as the experiments will
show.

The hash table is based on separate chaining.  Each chain is a sorted linked
list of elements.
The lists are immutable: to insert or delete an element, we use path copying to create a new copy of the
list, and change the hash table entry to point to the new copy using CAS.
For short chains, path copying yields a simple, efficient list
implementation that avoids the need for the mark bits on pointers that
are used in many lock-free lists.  
\hedit{In our experiments, we pick the size of the hash table so that the load factor is about $0.5$, so} chains are
very short on average.  An important property of this structure is that vCAS objects
never point indirectly to other vCAS objects---i.e., the head pointer of the
list is a vCAS object but none of the links are. 
The experiments run by B\"ottcher et al.~\cite{bottcher2019scalable} also had this property.
\hedit{The hashtable workload has similarities with many database workloads where the table entries are versioned but not the index~\cite{neumann2015fast,Wu17}.}

The chromatic tree is a  binary search tree 
with a lazy version of the red-black  balancing scheme.
Unlike the hash table, in a chromatic tree, or any concurrent 
tree, the vCAS objects \emph{do} point indirectly to others.  In particular,
a node points to its child, which can contain other vCAS objects.
This has a significant impact on some of the experimental results,
since not collecting a node soon enough means that its children will
not be collectible.  
\hedit{This represents an interesting aspect that does not show up in previous database MVGC work.}

Both data structures \hedit{store 32-bit integer keys and values} and support inserts, deletes, and \hedit{rtxs}.
On the trees we support range
queries---i.e., returning the keys within a range.  For hash tables we
support returning the value of multiple keys.  All operations are
linearizable.

\noindent  
{\bf Optimizations.}  
We use a backoff scheme in all implementations to reduce contention on  the global
timestamp counter. A process reads the counter
and if, after some variable delay, the counter has increased,
it simply uses the incremented value.  Otherwise, it updates the
counter using CAS.  When using version lists as in~\cite{WBBFRS21a},
reading an object requires first reading the location of
the head of its version list, and then reading the version the head points to.
As suggested in~\cite{WBBFRS21a}, we
avoid this level of indirection by \hedit{satisfying the \emph{recorded-once} property they define} and placing the timestamp and pointer
to the next older version in the object itself.

We tuned  parameters of the range tracking object~\cite{BBF+21},
including the number of lists to dequeue from the shared queue $Q$ in a flush  
and the number of elements per list.
When adding a list to $Q$, we omit already obsolete versions from it.

The authors of Steam \cite{bottcher2019scalable} suggest a heuristic 
of periodically scanning timestamps announced by \rtxs, instead of scanning every time a list is compacted.
This optimization does not preserve the $O(P)$ bound on the size of a version list \hedit{but it is crucial for performance}.
In \STEAMLF\, we \hedit{apply this optimization and} \y{we} scan announcements every 1ms.

\noindent
{\bf Workload.} In our experiments, we vary the following 
parameters:
(a) size of the multiversion data structure (denoted by $n$),
(b)~operation mix,
(c) size of \rtxs, (d) number of threads, and (e) the distribution from which keys are drawn.
In most experiments, we prefill each data structure with either $n=100K$ or $n=10M$ keys.
These sizes are chosen to illustrate performance when the data set fits and does not fit into the L3 cache.
We perform a mix of operations, consisting of inserts and deletes (done in equal numbers), as well as read transactions.
Keys for operations and the initial values of the data structure are drawn randomly from
the range $[1,2n]$, ensuring that the size of the data structure remains approximately $n$ throughout the experiment.


Our \rtxs\ 
search for all keys in the range $(a, a+s)$ where $a$ is drawn
uniformly at random from the range $[1, 2n-s]$ ($\mathit{s}$ is
the \rtx\ size).
The trees search by using the ordering while hash tables search by checking each individual key in the range.
For insert and delete operations, we draw keys from both the uniform distribution and Zipfian
distribution with parameter 0.99, which is the default in the YCSB
benchmark~\cite{YCSB}. 

By tuning the number of threads, we can vary the amount of contention, and also study the effects of oversubscription.
We see that these 
have a big impact on how MVGC algorithms perform since they lead to longer version lists.

\input{observations}

%% file: observations.tex

\input{tree-figures}



\subsection{Evaluation} 

Figures~\ref{fig:tree-100K-undersub-zipf}--\ref{fig:hashtable-100K-update-heavy} give
a cross section of our experimental results.  There was no qualitative difference between experiments on the uniform and Zipfian distributions, so we just include the Zipfian here.  More graphs are given in \er{\arxcam{Appendix \ref{additional-experiments}}{the full version \cite{arxiv}}}.
All experiments cover all five GC techniques.  Figures~\ref{fig:tree-100K-undersub-zipf}-\ref{fig:hashtable-100K-undersub-zipf} show \rtx\ and update throughput separately by placing \rtxs\ and updates on separate threads.  
This highlights the performance differences on updates, since \rtxs\ tend to be affected less.
\hedit{The workload in these figures consists of 40 update threads, 40 threads performing rtxs of size 16, and 40 threads performing variable-sized rtx. The throughputs of the variable-sized rtxs are shown in the leftmost graph of each figure.}
Figures~\ref{fig:tree-update-heavy} and~\ref{fig:hashtable-100K-update-heavy} \hedit{aggregate throughput by having each thread perform a mix of both types of operations}, which is more representative of a real workload.
\Eric{Is it worth saying here that same threads do all ops in contrast to previous sentence?  Isn't that what we mean when we say it is more representative of real workload?}\Hao{incorporated Eric's suggestion}
Based on the experiments, no one algorithm consistently outperforms the others; relative performance depends on various factors described below.  

\noindent {\bf Space.} We first consider space usage.  Perhaps the most interesting aspect of space is how poorly \STEAMLF{} performs on trees 
relative to hash tables.
This is explained by our previous discussion of vCAS objects indirectly pointing to other vCAS objects in the chromatic tree but not the hash table.  
\STEAMLF{} generally keeps around 1.5 to 2 versions per version list \hedit{(see} \y{tables of} \hedit{Figures~\ref{fig:tree-100K-undersub-zipf} and~\ref{fig:tree-10M-zipf-undersub})} because if there is an ongoing \rtx\ when a new version is added, the existing version is still needed by the \rtx\ and cannot be collected.  Soon after, when all current \rtxs\ finish, it becomes obsolete, but \STEAMLF{} will not  compact the list until that location is updated again. 
This is sometimes called the \emph{dusty corners problem}~\cite{larson2013hekaton} because corners (i.e., version lists) that are visited infrequently will be cleaned infrequently.
This delay only causes a space overhead factor of at most 2 when there is no indirection, but can be much worse with indirection.
In particular, the old version that is not collected can point to another node containing versioned objects, and that node and its versions will not be collected.  This effect can be seen for the chromatic trees in Figures \ref{fig:tree-100K-undersub-zipf} and \ref{fig:tree-100K-update-heavy}.  Even for small \rtxs, \STEAMLF{} performs much worse than the others.  
The experiments in the original Steam paper~\cite{bottcher2019scalable} used a \hedit{typical database implementation where only table entries are multiversioned, so versions do not indirectly point to one another and they did not see this issue as severely as we do.}

\EBR{} does not suffer from dusty corners since a version can be removed as soon as all active \rtxs\ at the time of an update have completed, although it has other problems discussed below.  The other three algorithms lack dusty corners since a version is removed as soon as the range tracker identifies  it, which can be much sooner than the next update.

The space usage of \EBR{} suffers badly for update-heavy workloads with large \rtxs\
and
with oversubscription.  
See, 
e.g., Figures ~\ref{fig:tree-10M-zipf-undersub},  ~\ref{fig:hashtable-100K-undersub-zipf}, and ~\ref{fig:tree-100K-update-heavy}.  
With large \rtxs\ the epochs become long,
and update-heavy workloads can create many versions
during an epoch that cannot be collected until a following epoch.  
Oversubscription can further prolong an epoch if a thread
in the epoch is delayed.  Indeed, this is exactly the problem
that removing intermediate versions is trying to solve.  
Thus, although \EBR\ does fine in many ``nice'' situations
it does very badly in adverse conditions.  
Figure ~\ref{fig:hashtable-100K-undersub-zipf}
indicates \EBR\ can use more than an order of magnitude
more memory than all the other MVGC schemes in such conditions.

A third observation is that \RTDLGC, \RTSLGC, and \bbf{} behave as expected in terms of space.  In particular, the theory indicates the memory is bounded by  three terms: 
the space needed for the current versions, a constant times the number of needed old versions,
and a function of the number of threads.
In Figure \ref{fig:tree-100K-update-heavy} for a small data structure of size 100K the third  term dominates as the number of threads increases for all three algorithms.
In Figure \ref{fig:tree-10M-update-heavy} for a larger data structure of size 10M the first two terms dominate making it hard to see the increase of space with more threads. 

In almost all cases \RTDLGC, \RTSLGC, and \bbf{} require less space than \EBR\ or \STEAMLF.  
  Importantly, although \EBR\ or \STEAMLF{} use slightly less memory than the other three in some cases, 
  the memory usage for the other three is much more predictable and never has any particularly bad cases.

Finally, we observe that \bbf{} almost always uses more memory than \RTDLGC{}, 
since in \bbf{} some nodes are obsolete but not collectable.
\RTDLGC{} almost always uses more memory than \RTSLGC\ because of the additional back pointers in the version lists, and also because \RTSLGC{} is able to preemptively splice out versions when traversing a list even  before they are returned by the range-tracker.
As discussed, the advantages of \RTSLGC{} come at the cost of losing time guarantees. 

\noindent {\bf Throughput.}  We now consider throughput.  Varying the GC algorithms affects \rtx\ and update throughput differently.  We first consider \rtx\ throughput.  The \rtx\ code is the same for all GC schemes, but there are at least two indirect effects.  

Firstly, a \er{scheme that compacts version lists less effectively tends} to have longer version lists, so \rtxs\ have to follow more version links.  This \er{explains} the very poor performance of very large \rtxs\ on the hash table with \EBR{} (Figure \ref{fig:hashtable-100K-undersub-zipf}). 


Secondly, \rtxs\ are running concurrently with  updates and, although they do not communicate with one another directly, there can be indirect effects, such as cache-line interference or faster updates increasing
the lengths of version lists.  It is hard to draw  general conclusions about these effects, but they
are  likely contributing to the minor variances.

The differences in update throughput among the algorithms is higher than for \rtxs.   This is to be expected since the code for updates differ significantly---e.g., some traverse whole version
lists, and some have to apply operations to the range-tracker data structure.  
Firstly, we observe that \bbf{} almost always does worse than all the others, often significantly.   This is the main motivation for designing the \RTDLGC{} and \RTSLGC{} algorithms.  There are a few data points where \bbf{} is faster than \RTSLGC{} but in these cases the \rtx\ throughput is
significantly faster for \RTSLGC{}.    
Secondly, \STEAMLF{} and \EBR{} mostly dominate the others, although the difference is usually small.    This is expected, given the extra cost of maintaining
the range-tracker data structure.
Thirdly, comparing our two new algorithms, \RTDLGC{} and \RTSLGC{}, in all graphs except Figure~\ref{fig:hashtable-100K-undersub-zipf}, the version lists were short enough that traversing from the beginning was faster than maintaining
back pointers and hence \RTSLGC{} performs better.
Even though the average version list length in Figure~\ref{fig:hashtable-100K-undersub-zipf} is small for \RTSLGC{}, the lists that get frequently updated tend to be long and \RTSLGC{} traverses an average of 17.9 version list nodes before reaching the one it wishes to splice out  



\noindent
{\bf Summary.} \hedit{Our evaluation shows that there is no definitively best MVGC scheme and we provide several insights into the types of workload and data structures each scheme is most suitable for.}
\er{While the performance of different schemes is often similar,
particularly in terms of throughput, there are a few important exceptions.}
\hedit{In particular, EBR uses up to 10 times more memory than the other schemes in workloads with long rtxs and Steam uses 8 times more memory in hierarchical data structures like trees. Our new schemes (SL-RT and DL-RT) are consistently faster than BBF+ and do not incur exceptionally high space usage in any workload. 
}

%% file: tree-figures.tex


\renewcommand{\arraystretch}{1.2}
\setlength{\tabcolsep}{3.3pt}

\begin{figure*}
	\begin{subfigure}{0.99\textwidth}
	\begin{tabular}{ l l }
	  \small \hspace{0.2cm} \textbf{Legend for Figures 4-8:}	& \includegraphics[width=0.55\linewidth,trim=0 0.5cm 0 0]{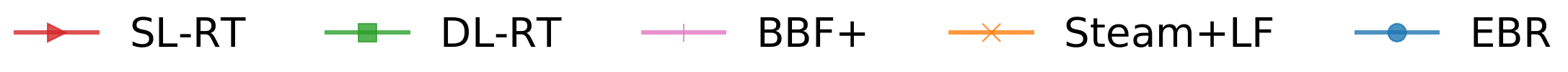}
	\end{tabular}
	\end{subfigure} \hfill
	\begin{subfigure}{0.74\textwidth}
		\includegraphics[width=0.99\linewidth,trim=0.2cm 0 0.25cm 1.5cm, clip]{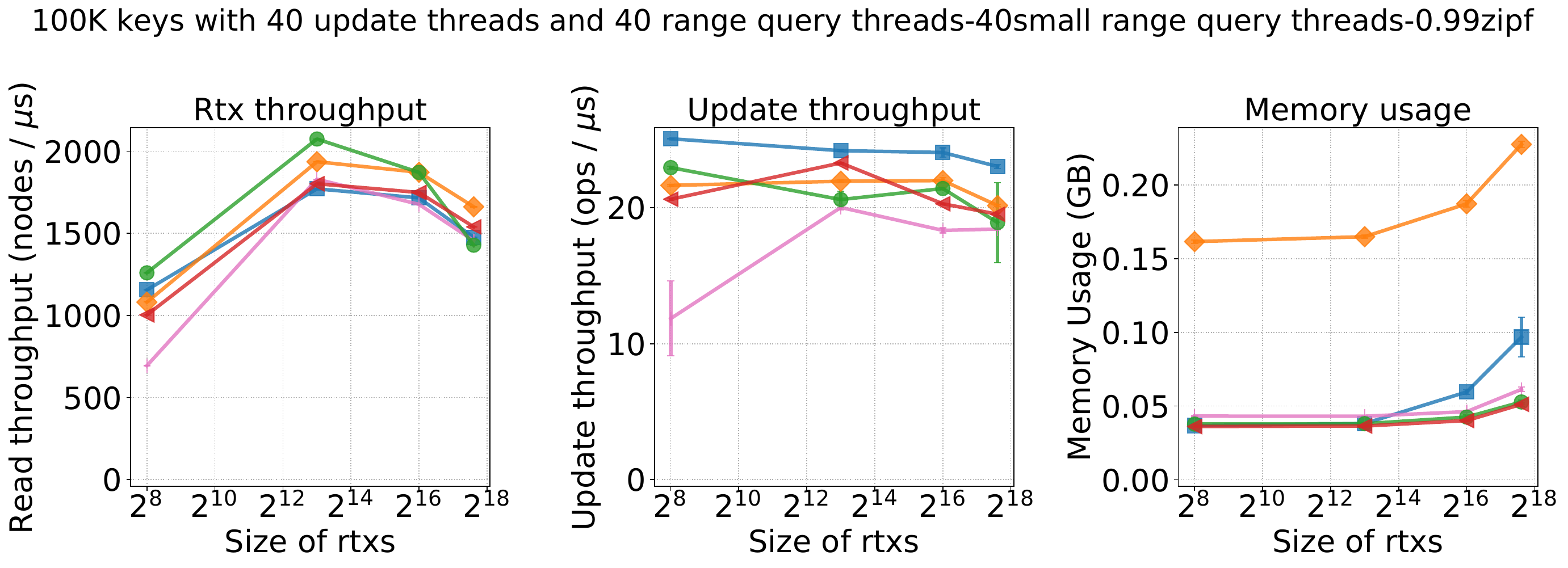}
	\end{subfigure} \hfill
	\begin{subfigure}{0.24\textwidth}
		\footnotesize
		\begin{tabular}{ |c|c|c|c|c| } 
			\multicolumn{5}{c}{\textbf{Average Version List Lengths}} \\
			\hline
			\begin{tabular}{@{}c@{}}\textbf{Size of} \vspace{-0.1cm} \\ \textbf{rtxs}\end{tabular}  & \textbf{$2^{8}$} & \textbf{$2^{13}$} & \textbf{$2^{16}$} & \textbf{$2^{18}$} \\ \hline \hline
			\textbf{SL-RT} & 1.03 & 1.05 & 1.15 & 1.39 \\ \hline
			\textbf{DL-RT} & 1.08 & 1.08 & 1.17 & 1.47 \\ \hline
			\textbf{BBF+} & 1.08 & 1.08  & 1.2 & 1.52 \\ \hline
			\textbf{Steam+LF} & 1.52 & 1.54 & 1.67 & 1.95 \\ \hline
			\textbf{EBR} & 1.05 & 1.05 & 1.28 & 2.32 \\ \hline
		\end{tabular}
	\end{subfigure}
	\caption{Tree with 100K keys, 40 update threads, 40 fixed-size rtx threads, 40 variable-size rtx threads.}
	\label{fig:tree-100K-undersub-zipf}
\end{figure*}

\begin{figure*}
	\begin{subfigure}{0.74\textwidth}
		\includegraphics[width=0.99\linewidth,trim=0.2cm 0 0cm 1.5cm, clip]{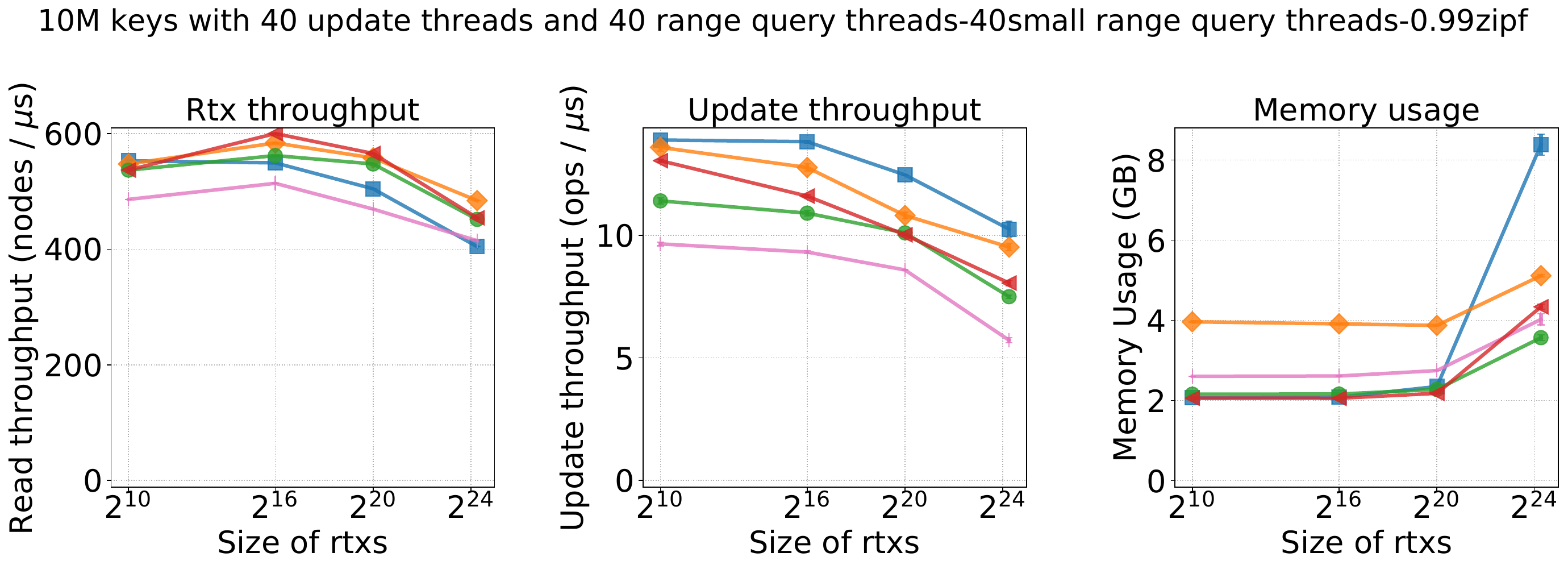}
	\end{subfigure} \hfill
	\begin{subfigure}{0.24\textwidth}
		\footnotesize
		\begin{tabular}{ |c|c|c|c|c| } 
			\multicolumn{5}{c}{\textbf{Average Version List Lengths}} \\
			\hline
			\begin{tabular}{@{}c@{}}\textbf{Size of} \vspace{-0.1cm} \\ \textbf{rtxs}\end{tabular}    & \textbf{$2^{10}$} & \textbf{$2^{16}$} & \textbf{$2^{20}$} & 20M \\ \hline \hline
			\textbf{SL-RT} & 1.0 & 1.0 & 1.04 & 1.42 \\ \hline
			\textbf{DL-RT} & 1.0 & 1.0 & 1.04 & 1.5 \\ \hline
			\textbf{BBF+} & 1.0 & 1.01  & 1.05 & 1.51 \\ \hline
			\textbf{Steam+LF} & 1.29 & 1.29 & 1.3 & 1.49 \\ \hline
			\textbf{EBR} & 1.0 & 1.01 & 1.11 & 2.32 \\ \hline
		\end{tabular}
	\end{subfigure}
	\caption{Tree with 10M keys, 40 update threads, 40 fixed-size rtx threads, 40 variable-size rtx threads.}
	\label{fig:tree-10M-zipf-undersub}
\end{figure*}

\begin{figure*}
	\begin{subfigure}{0.74\textwidth}
		\includegraphics[width=0.99\linewidth,trim=1.99cm 0 1.2cm 1.5cm, clip]{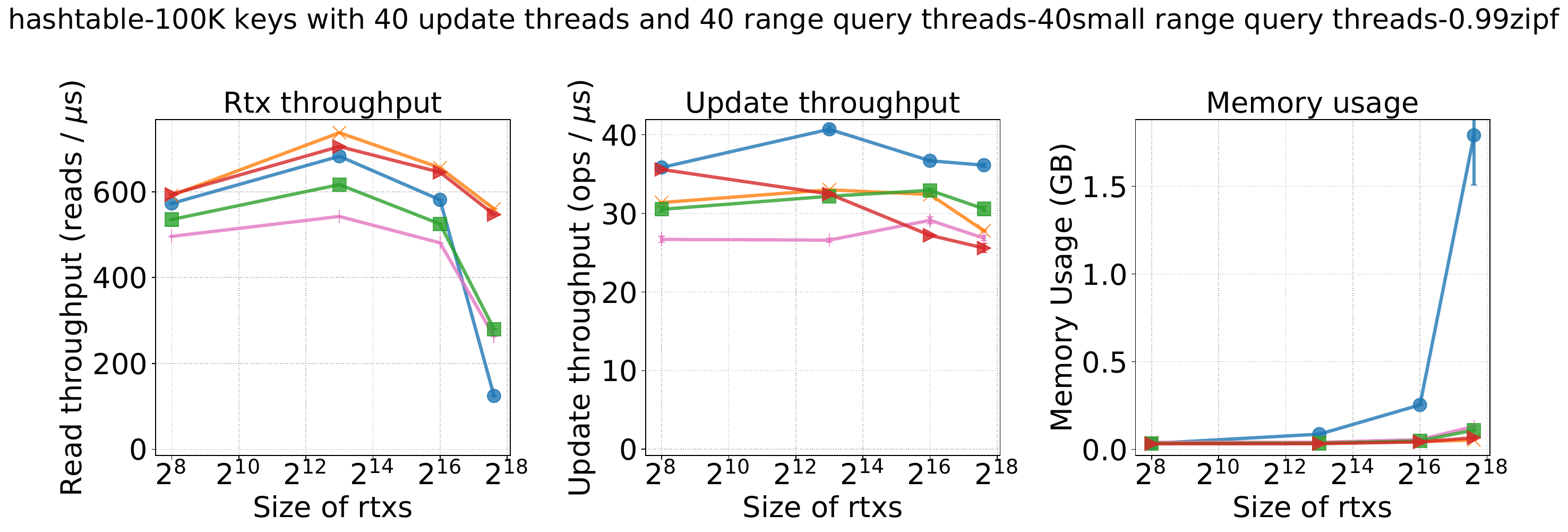}
	\end{subfigure} \hfill
	\begin{subfigure}{0.24\textwidth}
			\footnotesize
			\begin{tabular}{ |c|c|c|c|c| } 
				\multicolumn{5}{c}{\textbf{Average Version List Lengths}} \\
				\hline
				\begin{tabular}{@{}c@{}}\textbf{Size of} \vspace{-0.1cm} \\ \textbf{rtxs}\end{tabular}    & \textbf{$2^{8}$} & \textbf{$2^{13}$} & \textbf{$2^{16}$} & \textbf{$2^{18}$} \\ \hline \hline
				\textbf{SL-RT} & 1.02 & 1.06 & 1.34 & 2.14 \\ \hline
				\textbf{DL-RT} & 1.05 & 1.08 & 1.65 & 4.6 \\ \hline
				\textbf{BBF+} & 1.05 & 1.08  & 1.61 & 4.76 \\ \hline
				\textbf{Steam+LF} & 1.64 & 1.68 & 1.93 & 2.51 \\ \hline
				\textbf{EBR} & 1.01 & 1.09 & 2.26 & 73.63 \\ \hline
			\end{tabular}
	\end{subfigure}
	\caption{Hash table with 100K keys, 40 update threads, 40 fixed-size rtx threads, 40 variable-size rtx threads.}
	\label{fig:hashtable-100K-undersub-zipf}
\end{figure*}

\begin{figure*}
   \begin{subfigure}{0.49\textwidth}
		\centering
		\includegraphics[width=0.99\linewidth,trim=0.20cm 0 12.3cm 1.5cm, clip]{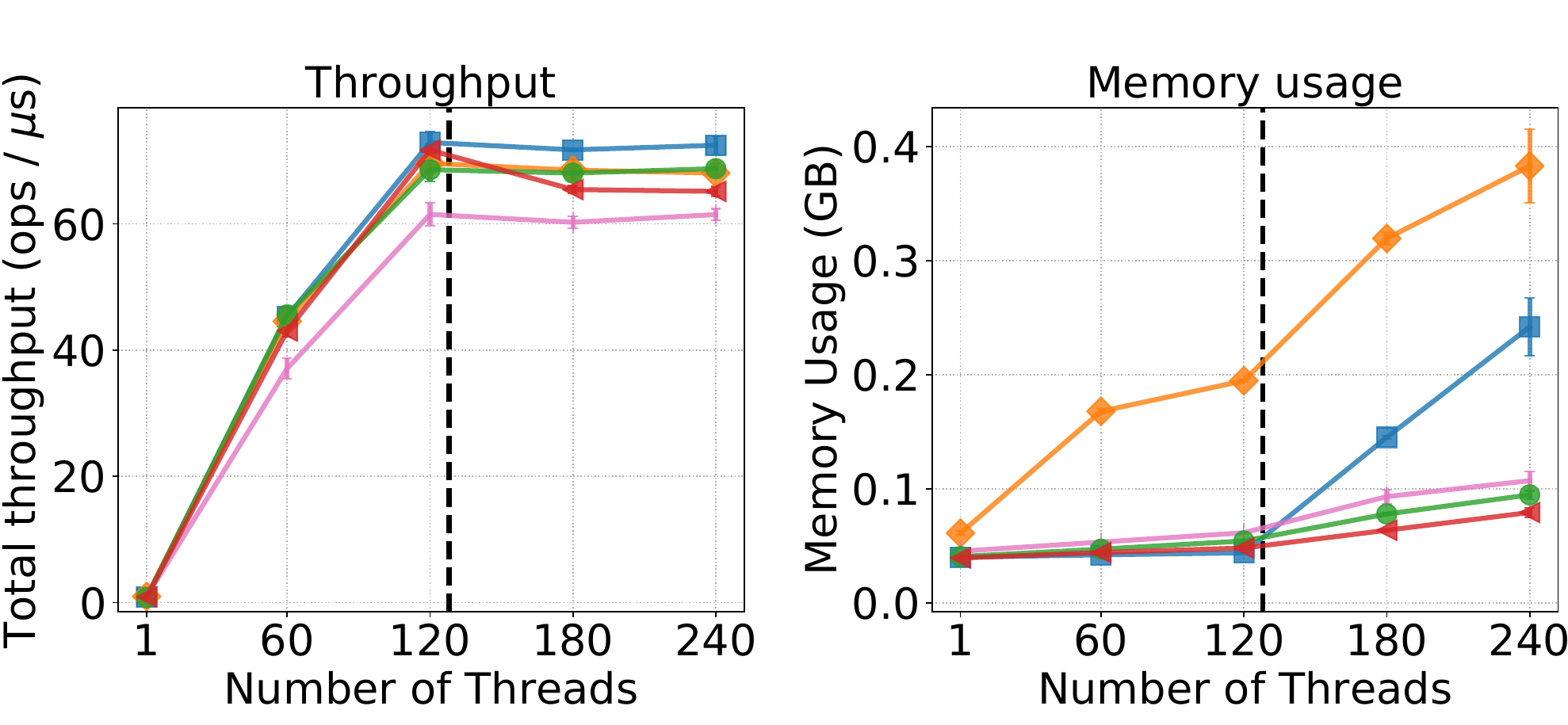}
		\caption{Tree with 100K keys}\label{fig:tree-100K-update-heavy}
	\end{subfigure} \hfill
	\begin{subfigure}{0.49\textwidth}
		\centering
		\includegraphics[width=0.99\linewidth,trim=0.20cm 0 12.3cm 1.5cm, clip]{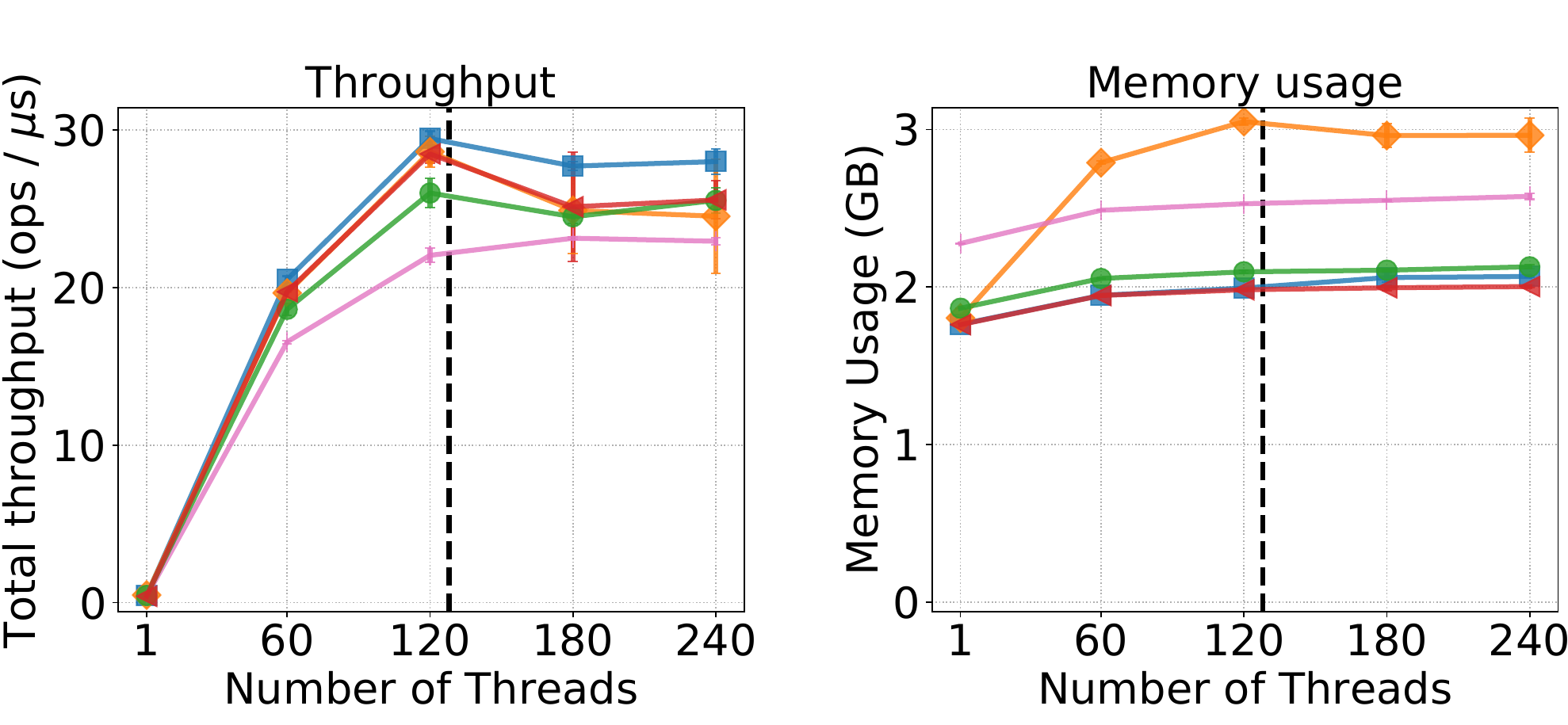}
		\caption{Tree with 10M keys}
		\label{fig:tree-10M-update-heavy}
	\end{subfigure}	
\caption{Workload with each thread performing 50\% updates, 49\% lookups, and 1\% read transactions of size 1024.}
\label{fig:tree-update-heavy}
\end{figure*}


\begin{figure}
	\includegraphics[width=0.99\linewidth,trim=0.25cm 0 12.45cm 2cm, clip]{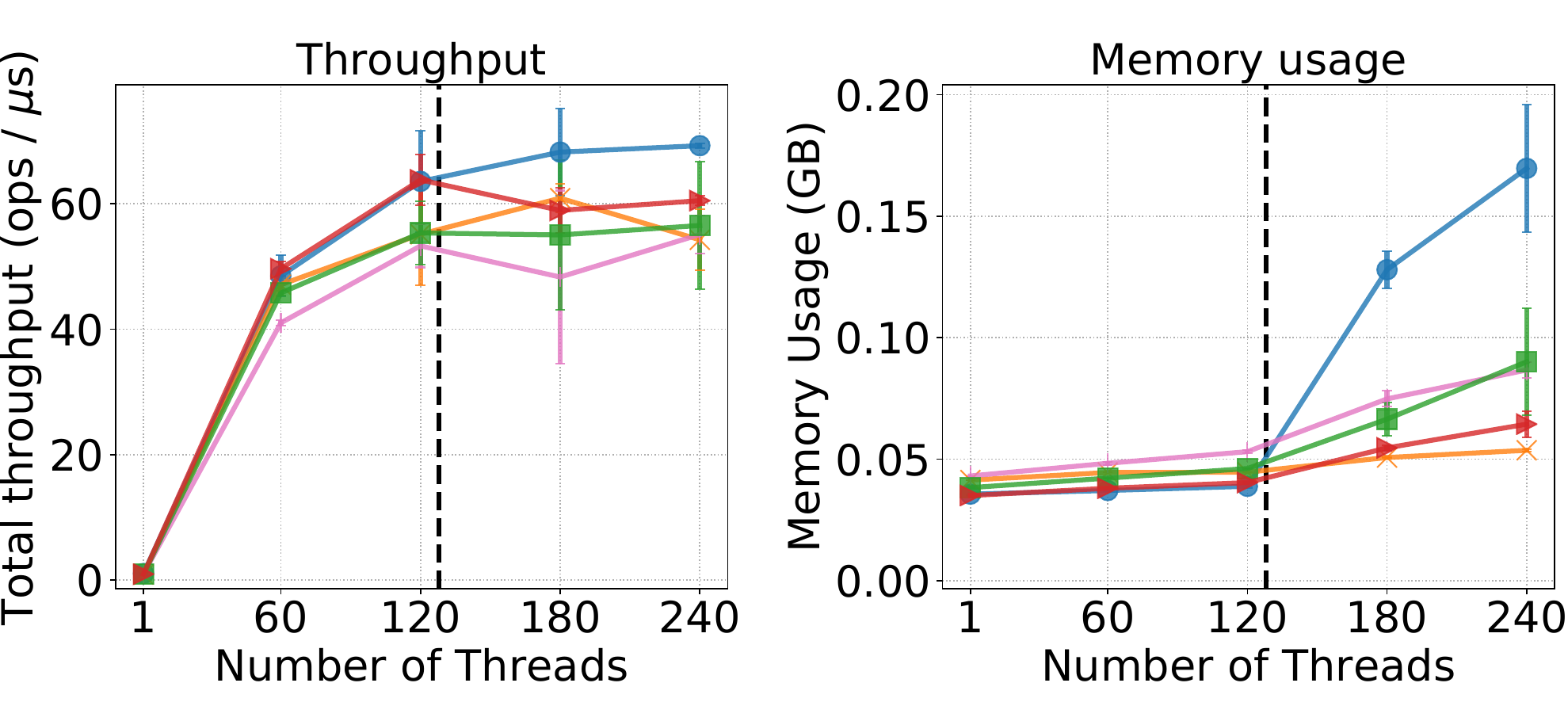}
	\caption{Hash table with 100K keys, each thread performs 50\% updates, 49\% lookups, and 1\% rtxs of size 1024.}
	\label{fig:hashtable-100K-update-heavy}
	\vspace{-0.1in}
\end{figure}
%

%% file: appendix.tex

\clearpage
\appendix

\section{Proofs for the Doubly-Linked List}
\label{DLL-appendix}

\here{I rewrote some the following claims to simplify proofs.
Doublecheck that all uses of this invariant in later proofs are still okay--may need to argue that antecedents of some claims are satisfied when the claim is used.  }

\resInvOrder*

\begin{proof}
Initially, the only node is the sentinel node and its \var{right} pointer is null, and there is no pending call to \remove, so the claim holds vacuously.
We show that every step preserves each of the three claims of the invariant.

\begin{enumerate}
\item
To prove that claim \ref{left-order} is preserved, we must only consider steps that can 
add a node to the list or change the \var{left}
pointer of a node.

When a node \vy\ is added to the list at line \ref{CAS-head} of a call to \var{\tryAppend(\vx,\vy)},
\var{\vy->left} has been set to \vx\ at line \ref{init-left}.
By the preconditions, this is the unique call of the form \var{\tryAppend(*,\vy)}, so
a pointer to \vy\ has not previously been written into shared memory.
So, no process could have updated \var{\vy->left} and it is therefore still \vx.
By the precondition of \tryAppend, \vx\ is a node that has previously been read from \var{head},
so it has been added to the list and $\vx\prec\vy$.
Since the CAS on \var{head} changes it from \vx\ directly to \vy, there are no nodes \vw\ satisfying
$\vx \prec \vw \prec\vy$.

The only instructions that update a \var{left} pointer are lines \ref{init-left} and \ref{CAS-left-DLL}.
As mentioned above, a node \vy\ has not been added to the list yet when line \ref{init-left} of \var{\tryAppend(\vx,\vy)}
updates \var{\vy->left}, so there is nothing more to prove.

Suppose an execution of line \ref{CAS-left-DLL} of a \remove(\vx) sets \var{\vy->left} to \var{v}.
Then, \var{v} and \vy\ are the values of the \remove's local variables \var{left} and \var{right}.
By induction hypothesis \ref{remove-order}, \var{v} has been added to the list and $\var{v} \prec \vx \prec \vy$ and every node \vw\ satisfying  
$\var{v} \prec  \vw \prec \vy$ is marked.
Thus, after the CAS, every node \vw\ satisfying $\var{y->left} \prec \vw \prec \vy$ is marked.

\item
To prove that claim \ref{right-order} is preserved, we must only consider steps that can 
change the \var{right}
pointer of a node.
The only instructions that set a \var{right} pointer to a non-null value are lines \ref{help-right}, \ref{set-right} and \ref{CAS-right}.

First, consider line \ref{set-right} of a \tryAppend(\vx, \vy), which sets \var{\vx->right} to \vy.
Since \var{head} was changed directly from \vx\ to \vy\ on line \ref{CAS-head},
$\vy$ has been added to the list and $\vx \prec \vy$, and
there are no nodes \vw\ satisfying $\vx \prec \vw \prec \vy$.

Next, consider an execution of line \ref{CAS-right} of a \remove(\vy) sets \var{\vx->right} to \vz.
Then, \vx\ and \vz\ are the values of the local variables \var{left} and \var{right}.
By the induction hypothesis, $\vx \prec \vy \prec \vz$ and every node \vw\ satisfying  
$\vx \prec \vw \prec \vz$ is marked.
Thus, after the CAS, every node \vw\ satisfying $\vx \prec \vw \prec \var{\vx->right}$ is marked.

Finally, consider a CAS step $s$ at 
line \ref{help-right} of \var{\tryAppend(\vy,*)} that changes \var{\vx->right} from null to \vy.
By the precondition of \tryAppend, \vy\ has been added to the list.
The \tryAppend\ read \vx\ from \var{\vy->left} at line \ref{append-read-left} before $s$.
Since \vx\ is not null at line \ref{help-right}, \vy\ cannot be the sentinel node, by induction hypothesis \ref{sentinel-left}.
So \vx\ is a node that has been added to the list and $\vx \prec \vy$ by induction hypothesis \ref{left-order}.
We prove the remainder of the claim by showing that there is no \vw\ such that $\vx \prec \vw \prec\vy$.

Since \var{\vx->right} is null in the configuration before $s$, it follows from induction hypothesis \ref{right-order} that 
\begin{equation}
\label{null-eqn}
\var{\vx->right}= \mbox{null at all times before }s.
\end{equation}
Node \vx\ cannot be marked when $s$ occurs:  before a \remove(\vx) could have marked
\vx, a \var{\tryAppend(\vx,*)} must have returned true (by the precondition for \remove),
and that \tryAppend\ would have set \var{\vx->right} to a non-null value at line \ref{set-right}, contrary to (\ref{null-eqn}).

To derive a contradiction, suppose there is some \vw\ such that $\vx \prec \vw \prec \vy$.
Consider the minimal such \vw\ with respect to the total order $\prec$.
(I.e., \vw\ is the successor of \vx\ in the order $\prec$.)
Since $\vx\leftarrow \vy$ at line \ref{append-read-left} before $s$,
$w$ is marked, by induction hypothesis \ref{left-order}.
This means that a \var{\remove(\vw)} has been invoked before $s$.
By the precondition of \remove, a \var{\tryAppend(\vw,$\vw'$)} returned true before $s$, for some $\vw'$.
Consider line \ref{append-read-left} of that \tryAppend.
When it occurs, \vx\ is unmarked (since it is still unmarked at $s$).
By induction hypothesis \ref{left-order}, $\var{\vw->left} \succeq \vx$ (since \vx\ is unmarked) 
and $\var{\vw->left} \preceq \vx$ (since \vx\ is the predecessor of \vw\ in the order $\prec$).
Thus, line \ref{append-read-left} reads \vx\ from \var{\vw->left}.
Line \ref{help-right} of that \var{\tryAppend(\vw,$\vw'$)} (which also occurs before $s$)
will perform a \var{CAS(\vx->right, null, \vw)}, violating (\ref{null-eqn}).
This contradiction proves that there is no \vw\ such that $\vx \prec \vw \prec \vy$.

\item
We prove that claim \ref{remove-order} is preserved by all updates to the local variables \var{left} and \var{right} in a \remove\ operation.

By the precondition of \var{\remove(\vy)}, \vy\ is not the sentinel node.  
Moreover, a call to \var{\tryAppend(\vy,*)} has returned true,
so by the preconditions of \var{\tryAppend}, we know that \vy\ has been added to the list
before the \remove\ begins.
Moreover, the \tryAppend\ executed line \ref{set-right}, after which \var{\vy->right} is non-null.
So, by induction hypothesis \ref{right-order}, \var{\vy->right} points to a node that has 
been added to the list.

So, when the local variable \var{left} in the \var{\remove(\vy)} operation
is initialized on line \ref{first-left} to the value read from \var{\vy->left},
the claim follows from induction hypothesis \ref{left-order}.
Similarly, when the local variable \var{right} is initialized on line \ref{first-right}, the
claim follows from induction hypothesis \ref{right-order}.

When the local variable \var{left} is advanced from \vx\ to $\vx'$ at line \ref{advance-left}
in a \var{\remove(\vy)},
we have $\vx' \leftarrow \vx$.
Since \vx\ is marked, so it is not the sentinel node, since the precondition to \remove\ ensures that the sentinel node is never marked.
By induction hypothesis \ref{left-order}, $\vx' \prec \vx$ and each node \vw\ satisfying
$\vx' \prec \vw \prec \vx$ is marked.
By induction hypothesis \ref{remove-order}, $\vx \prec \vy$ and each node \vw\ satisfying
$\vx \prec \vw \prec \vy$ is marked.
Moreover, by the test on line \ref{advance-left}, \vx\ is marked.
The claim follows.

A symmetric argument shows that advancing the local variable \var{right} at 
line \ref{advance-right} also preserves the invariant.

\item
Finally we show that claim \ref{sentinel-left} is preserved because no step can change
the \var{left} pointer of the sentinel node.
Line \ref{init-left} of \tryAppend\ changes the \var{left} field of a non-sentinel node, by the precondition to \tryAppend.
Consider a CAS step at line \ref{CAS-left-DLL} of \var{\remove(\vy)} that changes \var{\vz->left}.
By induction hypothesis \ref{remove-order}, $\vy \prec \vz$, so \vz\ cannot be the sentinel node
(which is the minimum element in the ordering $\prec$).
\end{enumerate}
\end{proof}

\resNoCrossovers*
\begin{proof}
To derive a contradiction, suppose the claim is false.

Since there is a node \vx\ between \vw\ and \vy\ in the total order $\prec$, 
the \var{head} pointer takes the value \vx\ between the times it takes the values \vw\ and \vy.
Thus, \vy\ is not added to the list
by a \var{\tryAppend(\vw,\vy)} operation, so the initial value of \var{\vy->left} is not \vw.
So, there is a CAS on line \ref{CAS-left-DLL} of a \remove\ operation that changes \var{\vy->left} to \vw.
This CAS was preceded by a test at line \ref{remove-test} that found \vy\ was not marked.
At that time, the local variables \var{left} and \var{right} of the \remove\ are equal
to \vw\ and \vy.  By Invariant \ref{inv-order}.\ref{remove-order}, \vx\ is marked since $\vw\prec\vx\prec\vy$.
Thus, \vx\ is marked before \vy.

An identical argument shows that a CAS at line \ref{CAS-left-DLL} of a \remove\ operation
changes \var{\vz->left} to \vx\ and at the execution of line \ref{remove-test} before
that CAS, \vx\ is not marked and \vy\ is marked, so \vy\ is marked before \vx.
This is a contradiction.
\end{proof}

\resCASAdvances*
\begin{proof}
To derive a contradiction, suppose $\vw \succ \vy$.
Consider the configuration $C$ when \vy\ is read from \var{\vz->left} on line \ref{left-right}
prior to the CAS.
At $C$, \vw\ and \vz\ are the values of the local variables 
\var{left} and \var{right} of the \remove.
By Invariant \ref{inv-order}.\ref{remove-order}, $\vw \prec \vz$.
In $C$, we have $\var{\vz->left} = \vy \prec \vw \prec \vz$.
By Invariant \ref{inv-order}.\ref{left-order}, \vw\ is marked in $C$.
This contradicts the fact that \vw\ is unmarked when line \ref{remove-test} is executed after $C$.
\end{proof}

\begin{figure}
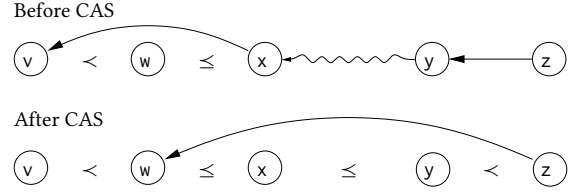
\caption{Proof of Lemma \ref{remove-only} and Lemma \ref{remove-only-SLL}.\label{fig-remove-only}}
\end{figure}

Recall that $\vx \lreach \vz$ if \vx\ can be reached from \vz\ by following \var{left} pointers.
More formally, $\vx \lreach \vz$ if (a) $\vx=\vz$ or (b) for some~\vy, $\vy\leftarrow \vz$ and $\vx\lreach \vy$.

\resRemoveOnly*
\begin{proof}
Refer to Figure \ref{fig-remove-only}.
Let \vx\ be the minimum node (with respect to $\prec$) such that $\vx \succeq \vw$ and $\vx \lreach \vy$ in $C$.
(This minimum is well-defined, since there exists a node \vx\ that satisfies both of these properties:
$\vy \succeq \vw$ by Lemma \ref{CAS-advances} and $\vy \lreach \vy$ holds trivially.)
Since $\vx \lreach \vy \leftarrow \vz$, we have $\vx \prec \vz$ by Lemma \ref{inv-order}.\ref{left-order}.

Let $\var{v}$ be the value of \var{\vx->left} in $C$.  By Invariant \ref{inv-order}.\ref{left-order},
$\var{v} \prec \vx$.
It must be the case that $\var{v} \prec \vw$; 
otherwise we would have $\var{v} \succeq \vw$ and $\var{v} \lreach \vy$, which would
contradict the fact that \vx\ is the \emph{minimum} node that satisfies these two properties.

To derive a contradiction, suppose $\vw \prec \vx$.  Then we have
$\var{v} \prec \vw \prec \vx \prec \vz$.  
We also have $\vw \leftarrow \vz$ after the CAS and $\var{v} \leftarrow \vx$ in $C$, which contradicts Lemma \ref{no-crossovers}.
Thus, $\vw \succeq \vx$.  By definition of \vx, $\vx \succeq \vw$ and $\vx\lreach\vy$ in $C$.
Thus, $\vx = \vw$ and $\vw\lreach\vy$.
\end{proof}

\resRemoveSafe*
\begin{proof}
A node can only be removed from the abstract list by a CAS step on a \var{left} pointer.
Consider a CAS at line \ref{CAS-left-DLL} that changes \var{\vz->left} from \vy\ to \vw.
By Invariant \ref{inv-order}.\ref{left-order} and Lemma \ref{remove-only}, only nodes \vx\ satisfying $\vw \prec \vx \prec \vy$ can
be removed from the abstract list by this CAS.
Since \vw\ and \vy\ are the values of the local variables \var{left} and \var{right}
in the \remove\ operation that performs the CAS, Invariant \ref{inv-order}.\ref{remove-order}
ensures that all such nodes \vx\ are marked when the CAS occurs.
\end{proof}

To ensure that we can linearize \remove\ operations, we show that a \var{\remove(\vx)} successfully removes \vx\ from the abstract list.
The \remove\ can only terminate when it performs a successful CAS on the \var{left} pointer of some
node appended to the list after \vx.  The following lemma shows that when that happens, \vx\ is
no longer in the abstract list.

\resRemoveSucceeds*
\begin{proof}
Since the \var{\remove(x)} terminates, it performs a successful CAS on a \var{left} pointer
at line \ref{CAS-left-DLL}.
Let \vw\ and \vz\ be the values of the \remove's local variables \var{left} and \var{right}
when the CAS occurs.
The CAS updates \var{\vz->left} to \vw.  By Invariant \ref{inv-order}.\ref{remove-order},
$\vw \prec \vx \prec \vz$.

To derive a contradiction, suppose that \vx\ is still in the abstract list in the configuration $C$
after this CAS.

Since \vx\ is in the abstract list and $\var{\vz->left} \prec \vx\prec\vz$, there must be nodes $\var{v}'$ and $\var{v}$ in the abstract list such that $\vx \lreach \var{v}' \leftarrow \var{v} \lreach \var{head}$ with $\var{v}' \prec \vz \prec \var{v}$.
Refer to Figure \ref{fig-remove-succeeds}.
More formally, let \var{v} be the minimum node (with respect to the order $\prec$) such that
$\var{v} \succeq \vz$ and $\var{v} \lreach \var{head}$ in $C$.
Let $\var{v}' = \var{v->left}$.  It follows from the definition of \var{v} that $\var{v}' \prec \vz$.
If \var{v} were equal to \vz, then since \vx\ is in the abstract list, \vx\ would have to be reachable
by following \var{left} pointers from \vw, which would violate \ref{inv-order}.\ref{left-order}, 
since $\vw\prec\vx$.  Thus, $\var{v} \succ \vz$.

So, we have $\vw \prec \vx \preceq \var{v}' \prec \vz \prec \var{v}$ and $\vw \leftarrow \vz$ and $\var{v}' \leftarrow \var{v}$,
which violates Lemma \ref{no-crossovers}.
\end{proof}

\begin{figure}
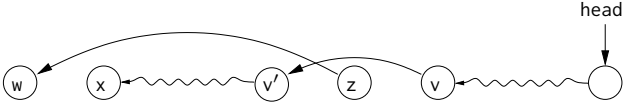
\caption{Proof of Lemma \ref{remove-succeeds}.\label{fig-remove-succeeds}}
\end{figure}

We now specify the linearization points of update operations
so that we maintain the invariant that, for all configurations $C$, 
the abstract list $L_C$ in that configuration is exactly the list that would be obtained
by performing the operations linearized before $C$ in their linearization ordering.
We linearize a successful \tryAppend\ operation (i.e., one that returns \var{true}) when it updates
the \var{head} pointer at line \ref{CAS-head}.
We linearize a \var{\remove(x)} when a CAS at line \ref{CAS-left-DLL} removes \var{x} from the 
abstract list.  
This step may be performed by the \remove\ itself or by another \remove\ operation.
By Lemma \ref{remove-only}, \var{n} will remain out of the 
abstract list forever once it is removed.
By Lemma \ref{remove-safe}, only marked nodes are removed from the abstract
list when a \var{left} pointer is updated.
By Lemma \ref{remove-succeeds}, there is a well-defined linearization point for each \remove\ operation
that terminates, and it is before the \remove\ terminates.  By Lemma \ref{remove-safe}, it is also
after the \remove\ operation performs line \ref{mark}, so the linearization point is during the
execution interval of the \remove.
Thus, the linearization is consistent with the abstract list.

A CAS of a \var{left} pointer may remove several nodes from the abstract list, 
and therefore be used as the linearization point of several \remove\ operations.
When linearizing such a batch of \remove\ operations, we order them in decreasing order by $\prec$.
In other words, we remove the batch of nodes from right to left.
Thus, when a CAS on a \var{left} pointer splices several nodes out of the abstract list,
we can think of the abstract list as undergoing a sequence of changes
by performing those \remove\ operations one by one.
This will be important, because it is sometimes necessary to linearize
searches in the \emph{middle} of such a batch of \remove\ operations.

We now describe how to linearize \search\ operations.
It is a common technique to argue that each node visited by a \search\ \emph{was}
in the data structure at some time during the search (e.g., \cite{EFR10,ORV10}).
Such proofs often rely on the fact that once a node is removed from the data
structure, its pointers can no longer be changed.  This does not hold for our
data structure, but we can still establish the required property using Lemmas 
\ref{no-crossovers} and \ref{remove-only}.
(An alternative approach would be to use the forepassed condition defined in \cite{FKE20}.)

\resSearchInv*
\begin{proof}
Consider any \var{\search(k)} operation in the execution.
We prove the claim holds for every configuration $C$ during the \search\ by induction 
on the number of steps the \search\ has performed.

For the base case, consider the configuration $C$ after the \search\ is invoked and 
executes line \ref{search-begin}, \vx\ is initialized to \var{head}, 
so it is the first node in the abstract list in $C$.

For the induction step, we must just verify that steps that modify the local variable \vx\ preserve the claim.
Consider an execution of line \ref{search-advance} that advances the \search\ from node \vx\ to $\vx'$,
and let $C$ be the configuration after that step.
Then, in $C$, $\vx' \leftarrow \vx$.
By the test at line \ref{search-test}, $\var{\vx->key} > \var{k}$.

First, consider the case where \vx\ is in $L_C$.  
Since $\vx' \leftarrow \vx$ in $C$, $\vx'$ is also in $L_C$ 
and its predecessor \vx\ in $L_C$ has a key greater than \var{k}.
Thus, the claim is satisfied at $C$.

So, for the remainder of the proof, assume \vx\ is not in $L_C$.
By the induction hypothesis, there was an earlier time during the \search\ when \vx\ was in the abstract list.
So there was a CAS step $s$ during the \search\ that removed \vx\ from the abstract list.

Lemma \ref{remove-only} implies that if $\vx' \lreach \vx$ after some CAS on a \var{left} pointer,
then $\vx' \lreach \vx$ before that CAS too.  Since $\vx' \lreach \vx$ in $C$, it follows that
$\vx' \lreach \vx$ holds in all configurations between the addition of \vx\ to the abstract list
and $C$.  In particular, this means that $\vx' \lreach \vx$ in the configuration before $s$.
So $\vx'$ is in the abstract list in that configuration.
We consider the two possible cases.

\begin{figure}
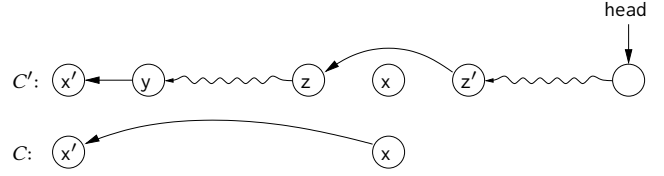
\caption{Proof of Lemma \ref{search-inv}.\label{fig-search-inv}}
\end{figure}

\begin{itemize}
\item
Suppose $s$ removes both $\vx$ and $\vx'$ from the abstract list.
Then $s$ removes some sequence of nodes from the abstract list by updating the \var{left} pointer
of some node \vz.
Since we linearize a batch of \remove\ operations from right to left, \vz\ is the predecessor
of $\vx'$ in the abstract list just before the linearization of \var{remove($\vx'$)}.
It follows from Lemma \ref{remove-only} that $\vx \lreach \vz$ in the configuration before $s$.
By the precondition that the list is always sorted, we have $\var{\vz->key} \geq \var{\vx->key} >\var{k}$.
Thus, the claim is satisfied just before the linearization point of the \var{remove($\vx'$)}.

\item
Suppose $\vx'$ is still in the abstract list in the configuration $C'$ following $s$.
If $\vx'$ is the first element of $L_{C'}$, then the claim holds.

Otherwise, there is some \vy\ such that in $C'$, $\vx' \leftarrow \vy \lreach \var{head}$.
To derive a contradiction, suppose $\vy \prec \vx$.
Then there must be nodes \vz\ and $\vz'$ such that $\vy \preceq \vz \prec \vx \prec \vz'$ and
$\vy \lreach \vz \leftarrow \vz' \lreach \var{head}$ in $C'$.  (See Figure \ref{fig-search-inv}:
at some point in the path from \var{head} to \vy, there must be consecutive nodes $\vz\leftarrow \vz'$
such that $\vz \prec \vx \prec \vz'$ since \vx\ is not in the abstract list.)
So we have $\vx' \prec \vx \prec \vx \prec \vz'$ and $\vz\leftarrow\vz'$ in $C'$ and $\vx' \leftarrow \vx$ in $C$.
This violates Lemma \ref{no-crossovers}.
Hence, $\vy \succeq \vx$.  We cannot have $\vy=\vx$, since \vy\ in in $L_{C'}$ and \vx\ is not.
So, $\vy \succ \vx$.  By the precondition that nodes are appended with non-decreasing keys,
this means that $\var{\vy->key} \geq \var{\vx->key} > \var{k}$.
So,  $\vx'$ is in $L_{C'}$ and its predecessor \vy\ has a key greater than $k$, as required.
\end{itemize}
\end{proof}

Suppose a \var{\search(k)} returns the value of a node \vx.
By the exit condition on line \ref{search-test}, $\var{vx->key} \leq \var{k}$.
Lemma \ref{search-inv} says there is a time during the \search\ when
\vx\ is in the abstract list and its predecessor in the list has a key greater than \var{k}.
Thus, at that time, \vx\ is the first node in the abstract list whose key is less than or equal to \var{k}.
We choose that time as the linearization point of the \search.

\section{Correctness of the Singly-Linked List}
\label{SLL-appendix}

We assume that the following preconditions are satisfied.
\begin{enumerate}[leftmargin=3mm,labelwidth=0mm,labelsep=1mm,topsep=0mm]
\item
If \var{\tryAppend(\vx,\vy)} is invoked, then \vx\ has been read from \var{head}, and
no other \var{\tryAppend(*,\vy)} has been invoked.
\item
\label{append-order}
If a node \vx\ is appended to the list before \vy, $\vx.\ts \leq \vy.\ts$.
\item
\label{searches-announced}
While a \var{\search(t)} is in progress, \vt\ appears in \Announce.
\item
\label{announce-pre}
At all times after a \var{compact(A,t,*)} is invoked and for all \var{i}, either $\Announce[\var{i}]\in \vA$
or $\Announce[\var{i}] \geq \vt$.
\end{enumerate}
These conditions are satisfied naturally in MVGC applications.
Since \compact\ takes its parameters from an AnnScan object (\vA,\vt),
\var{scanAnnounce} ensures that \vt\ was copied from shared memory before \vA.
Since the global timestamp is non-decreasing,
any value stored in \var{Announce[i]} after that entry was copied into \vA
will be greater than \vt.  (We assume \rtxs\ announcing a timestamp
copy the global timestamp into \var{Announce} atomically, as described in~\cite{WBBFRS21a}.
An optimization that avoids this is described in Appendix \ref{avoid-copy}.)

As in Section \ref{DLL-correctness}, 
we say a node is appended to the list when a pointer to the node is stored in \var{head}.
We consider the sentinel node to be appended to the list when the list is created.
We also use the notation $\vx\leftarrow\vy$ and $\vx\prec\vy$ as in Section \ref{DLL-correctness}.
We first prove some simple facts about the algorithm.

\begin{invariant}
\label{sll-order}\mbox{ }
\begin{enumerate}
\item
If \vy\ has been appended to the list and is not the sentinel node and $\vw\leftarrow \vy$  then \vw\ is a node that has been appended to the list and $\vw \prec \vy$.
\item
If a CAS at line \ref{CAS-left-SLL} changes \var{\vy->left} from \vx\ to \vw, then $\vw \prec \vx$.
\item
A \compact\ never sets 
\var{cur}, \var{next} or \var{newNext}  to null.
\end{enumerate}
\end{invariant}

\begin{proof}
We show that each step preserves the invariant.  Assume the invariant holds prior to the step.
We need only consider steps that append a node, successfully CAS a \var{left} pointer and update the local variables \var{cur}, \var{next} and \var{newNext}.

First, consider a step that appends a node \vy\ to the list (line \ref{CAS-head-SLL}).
If $\vw \leftarrow \vy$ when \vy\ is appended to the list by a CAS at line \ref{CAS-head-SLL}, 
then \var{head} is changed from \vw\ directly to \vy.  
Thus, \vw\ is a node that was appended to the list before \vy.

The only step that changes a \var{left} pointer after it is initialized is a successful
CAS at line \ref{CAS-left-SLL}.
Suppose this step changes \var{\vy->left} from \vx\ to \vw.
Then, \vx\ and \vw\ are the values of local variables \var{next} and \var{newNext} when the CAS occurs.
Prior to the CAS, \vw\ has been reached from \vy\ by following \var{left} pointers,
first at line \ref{init-next} and then at lines \ref{init-new-next}--\ref{advance-new-next}. 
Moreover, before advancing from one node to the next by reading a \var{left} pointer (at line \ref{init-next}, \ref{init-new-next} or \ref{advance-new-next}),
the algorithm first checks (at line \ref{test-sentinel}, \ref{test-next-needed} or \ref{test-new-next-obs}, respectively)
that the former node is not the sentinel,
either explicitly, or by checking that the node's \ts\ field is greater than \var{A[i]}.
Since the invariant holds before the CAS,  \vw\ has been appended and $\vw \prec \vx \prec \vy$.

As mentioned in the previous paragraph, when we set \var{next} at line \ref{init-next}
or set \var{newNext} at lines \ref{init-new-next} or \ref{advance-new-next} by reading
the \var{left} field of a node, we first check that that node is not the sentinel.
Since the invariant holds before reading this \var{left} field, the value we store in \var{next}
or \var{newNext} is non-null.
Similarly, if \var{next} is updated at line \ref{reread-left}, the test at line \ref{test-sentinel}
guarantees that \var{cur} is not the sentinel, so its left pointer is non-null.

We check that all steps that set the \var{cur} variable give it a non-null value.
The \var{cur} variable is set to the value read from \var{head} at line \ref{init-cur}, and
\var{head} is never null.
Lines \ref{advance-cur-1} and \ref{advance-cur-2} copy the value \var{next} into \var{cur},
and this value is non-null, by the induction hypothesis.
Line \ref{advance-cur-3} sets \var{cur} to the value read from the \var{left} field of a node that
is known not to be the sentinel (by the test at line \ref{test-sentinel}), so its \var{left}
field is non-null by the induction hypothesis. 
\end{proof}

It follows from Invariant \ref{sll-order} that all pointers dereferenced in \compact\ are non-null.

The following technical lemma shows that \compact\
updates \vi\ appropriately to carry out the test of whether node \var{next}
is needed.

\begin{lemma}
\label{merge-lemma}
If the test on line \ref{test-next-needed} fails, then
either \vi\ is the index of the last entry of \vA\ or 
$\var{cur->\ts} \leq \var{\vA[i+1]}$.
\end{lemma}

\begin{proof}
Let $C$ be the configuration after the test on line \ref{test-next-needed} fails.
Consider the configuration $C_i$ after the last time the local variable \vi\ was updated prior to  $C$.
If this was when \vi\ was initialized at line \ref{init-i}, then \vi\ is the index of the last element of \vA, which makes the claim true.
Otherwise, \vi\ was updated at line \ref{advance-i}.
Then, $\var{A[i+1]} \geq \var{cur->\ts}$ held at $C_i$.
Between $C_i$ and $C$, we argue that changes to \var{cur} can only cause \var{cur->\ts} to decrease.
If \var{cur} is updated at line \ref{advance-cur-1} or \ref{advance-cur-2}
it is changed to the value read from the \var{left} pointer of the previous \var{cur} node at line \ref{init-next}.  Similarly, line \ref{advance-cur-3} updates \var{cur} to the value read
from the \var{left} pointer of the previous \var{cur} node.
By Lemma \ref{sll-order} and precondition \ref{append-order},
updating \var{cur} between $C_i$ and $C$ by traversing \var{left} pointers cannot
cause \var{cur->ts} to increase.
Thus, at $C$, we still have $\var{cur->ts} \leq \var{\vA[i+1]}$.
\end{proof}

We say that $(\vA,\vt)$ \emph{is written before} $(\vA',\vt')$ 
if an \AnnScan\ object $(\vA,\vt)$ is stored in \GlobalAnnScan\ before $(\vA',\vt')$.
The following lemma is proved using the fact that creating the earlier copy
completes before creating the later copy begins, because of the way
we use CAS to update \GlobalAnnScan.
This, in turn, guarantees that once a node becomes unneeded, it is never needed 
by a later \AnnScan\ pair.

\begin{lemma}
\label{order-copies}
Suppose $(\vA_1,\vt_1)$ is written before $(\vA_2,\vt_2)$. 
For any node \vx, if \vx\ is $needed(\vA_2,\vt_2)$ then \vx\ is $needed(\vA_1,\vt_1)$.
\end{lemma}

\begin{proof}
Since every CAS on \GlobalAnnScan\ tries to install a newly created object, there is no ABA problem
on \GlobalAnnScan.  Thus, \GlobalAnnScan\ is not updated between the time it is read
on line \ref{read-globalannscan} and the time a successful CAS is applied to it on line \ref{CAS-globalannscan}.
Thus, $(\vA_1,\vt_1)$ is stored in \GlobalAnnScan\ before the process that stores $(\vA_2,\vt_2)$
begins reading the values $(\vA_2,\vt_2)$ from the \Announce\ array and the global timestamp.
Thus, $t_1\leq t_2$. 
Moreover, if some value $\vt < \vt_1$ appears in $\vA_2$, then it must also appear in $\vA_1$
since we assume values written into the \Announce\ array are atomically copied from the global
timestamp.

Assume \vx\ is $needed(\vA_2,\vt_2)$.  We consider several cases.

If $\vx.\ts > \vt_1$, then \vx\ is also $needed(\vA_1,\vt_1)$, by definition.

If $\vx.\ts \leq \vt_1$ and \vx\ is the last appended node with timestamp at most $\vt_2$,
then it is also the last appended node with timestamp at most $\vt_1$ since $\vt_1 \leq \vt_2$.
Thus, \vx\ is $needed(\vA_1,\vt_1)$.

Otherwise, $\vx.\ts \leq \vt_1$ and \vx\ is the last appended node with timestamp at most $\vA_2[\vi]$
for some \vi.
If $\vA_2[\vi] \geq \vt_1$, then \vx\ is also the last appended node with timestamp at most $\vt_1$,
and is therefore $needed(\vA_1,\vt_1)$.
If $\vA_2[\vi] < \vt_1$ then the value $\vA_2[\vi]$ must also appear in $\vA_1$.
So, again \vx\ is $needed(\vA_1,\vt_1)$.
\end{proof}

The following key lemma describes how a \var{\compact(\vA,\vt,*)} works:
it traverses nodes that are $needed(\vA,\vt)$ and splices out those nodes that are not $needed(\vA,\vt)$.
The proof is necessarily quite technical because, even though the set of needed nodes 
can only be reduced when \GlobalAnnScan\ is updated, there many \compact\ routines
simultaneously traversing a version list, whose arguments are $(\vA,\vt)$ pairs read from 
\GlobalAnnScan\ at different times--some out of date and some more current.

\begin{lemma}
\label{remove-unneeded}
If an invocation of \var{\compact(\vA,\vt,\vh)} sets its local variable \var{cur}
to a node $\vx\neq\vh$, then \vx\ is $needed(\vA,\vt)$.
If an invocation of \var{\compact(\vA,\vt,\vh)} performs a successful CAS that stores
\var{v} in $\var{\vx->left}$, then  \vx\ is $needed(\vA,\vt)$ (unless $\vx=\vh$) and \var{v} is $needed(\vA,\vt)$ and
for all \vw\ satisfying $\var{v} \prec \vw \prec \vx$, \vw\ is not $needed(\vA,\vt)$.
\end{lemma}
\begin{proof}
We prove the lemma by induction on steps:  we assume the lemma is true in a prefix of
an execution and show that it remains true when an additional step $s$ is appended to the prefix.
We first prove a couple of technical claims that will be used several times in the induction step.

\begin{claim}
\label{needed-claim}
If $\vx \leftarrow \vz$ at some time before $s$ and either 
$\vx.\ts \leq \vA[\vj] < \vz.\ts$ for some 
\vj\ or $\vz.\ts > \vt$, then \vx\ is $needed(\vA,\vt)$.
\end{claim}

\begin{proof}[Proof of Claim]
We consider several cases.
\begin{enumerate}
\item
If $\vx.\ts > \vt$, then \vx\ is $needed(\vA,\vt)$ by definition.
\item
If \vx\ is the last appended node whose timestamp is at most \vt, then \vx\ is $needed(\vA,\vt)$ by definition.
\item
\label{non-last}
Suppose $\vx.\ts \leq \vt < \vz.\ts$ but \vx\ is not the last node whose timestamp is at most \vt.
Let \vy\ be the last appended node whose timestamp is at most \vt.
By definition, \vy\ is $needed(\vA,\vt)$.
Since $\vz.\ts > \vt$, we have
$\vx \prec \vy \prec \vz$ by precondition \ref{append-order}.
By the hypothesis of the claim, some invocation \var{\compact($\vA'$,$\vt'$)} set 
\var{\vz->left} to \vx\ before $s$.
By the induction hypothesis, 
\vx\ is $needed(\vA',\vt')$ but \vy\ is not $needed(\vA',\vt')$.
Since \vy\ is $needed(\vA,\vt)$ but not $needed(\vA',\vt')$,
$(\vA,\vt)$ is written before $(\vA',\vt')$, by Lemma \ref{order-copies}.  
Since \vx\ is $needed(\vA',\vt')$, \vx\ is also $needed(\vA,\vt)$, again by Lemma \ref{order-copies}.
\item
If \vx\ is the last appended node whose timestamp is at most \vA[\vj] (for some \vj), then \vx\ is $needed(\vA,\vt)$, by definition.
\item
Otherwise, $\vx.\ts \leq \vA[\vj] < \vz.\ts$ for some \vj, but \vx\ is not the last appended node with timestamp at most $\vA[\vj]$.  This case is very similar to Case \ref{non-last}.
Let \vy\ be the last appended node whose timestamp is at most $\vA[\vj]$.
By definition, \vy\ is $needed(\vA,\vt)$.
Since $\vz.\ts > \vA[\vj]$, we have
$\vx \prec \vy \prec \vz$ by precondition \ref{append-order}.
By the hypothesis of the claim, some invocation \var{\compact($\vA'$,$\vt'$)} set \var{\vz->left} to \vx\ before $s$.
By the induction hypothesis, 
\vx\ is $needed(\vA',\vt')$ but \vy\ is not $needed(\vA',\vt')$.
Since \vy\ is $needed(\vA,\vt)$ but not $needed(\vA',\vt')$,
$(\vA,\vt)$ is written before $(\vA',\vt')$, by Lemma \ref{order-copies}.  
Since \vx\ is $needed(\vA',\vt')$, \vx\ is also $needed(\vA,\vt)$, again by Lemma \ref{order-copies}.
\end{enumerate}
\end{proof}

\begin{claim}
\label{CAS-needed}
If $s$ is an execution of line \ref{CAS-left-SLL} or \ref{advance-cur-3},
let \var{v} and \vx\ be the nodes stored in \var{newNext} and \var{cur}, respectively, in the configuration before $s$.
Then $\var{v} \prec \vx$ and \var{v} is $needed(\vA,\vt)$ and for all \vw\ satisfying $\var{v} \prec \vw \prec \vx$,
\vw\ is not $needed(\vA,\vt)$.
\end{claim}

\begin{proof}[Proof of Claim]
The node \var{v} was reached from \vx\  by following \var{left} pointers,
at line \ref{init-next} and then at lines \ref{init-new-next}--\ref{advance-new-next}.
By Invariant \ref{sll-order}, $\var{v}\prec \vx$.
Let $\vx=\vw_1,\vw_2,\ldots,\vw_k=\var{v}$ be the sequence of nodes traversed to get from \vx\ to \var{v}.
By the tests on lines \ref{test-next-needed} and \ref{test-new-next-obs}, $\vw_{k}.\ts \leq \vA[\vi] < \vw_{k-1}.\ts$.
Since $\vw_k$ was read from $\vw_{k-1}$ before $s$, Claim \ref{needed-claim} implies that
$\var{v}=\vw_k$ is $needed(\vA,\vt)$.

It remains to show that all nodes \vw\ satisfying $\var{v} \prec \vw \prec \vx$ are not $needed(\vA,\vt)$.
Since the test at line \ref{cur-too-new} failed, $\vx.\ts \leq \vt$.
So a node $\vw \prec\vx$ cannot be the last appended node with timestamp at most $\vt$,
nor can \vw\ have a timestamp greater than $\vt$.
To complete the proof that \vw\ is not $needed(\vA,\vt)$, we must show that it is not the
last appended node with timestamp at most \var{\vA[\vj]} for any \vj.
We do this in two parts, first considering $\vj > \vi$ and then $\vj \leq \vi$.
(Recall that $\var{A[0]}=-1$.)

Since \vx\ is the value of \var{cur} when the test at line \ref{test-next-needed} fails,
it follows from Lemma \ref{merge-lemma} that 
either $\vi$ is the index of the last element of $\vA$ or $\var{\vx->\ts} \leq \vA[\vi+1]$.
Thus, any node $\vw \prec\vx$ cannot be the last appended
node whose timestamp is at most $\vA[j]$ for any $\vj > \vi$, since then
\vi\ would not be the last index and so $\vx.\ts \leq \vA[\vi+1] \leq \vA[\vj]$.

To derive a contradiction, suppose that 
there is some \vw\ such that $\var{v} \prec \vw\prec \vx$ and
\vw\ is $needed(\vA,\vt)$.
Since we have already eliminated all other possibilities, this means that
for some $j\leq i$, 
\vw\ is the last appended node whose timestamp is at most $\vA[\vj]$.
Since $\vw.\ts \leq \vA[\vj] \leq \vA[\vi] < \vw_{k-1}.\ts$, we must have $\var{v} = \vw_k \prec \vw \prec \vw_{k-1}$.
(In particular, this means that $\vw_{k-1}$ cannot be $needed(\vA,\vt)$.)
Since $\vw_k$ was read before $s$ from $\vw_{k-1}$\var{->left},
the induction hypothesis implies that
some \var{\compact($\vA'$,$\vt'$)} did a CAS that set $\vw_{k-1}$\var{->left} to $\vw_k$ and
$\vw_{k-1}$ is $needed(\vA',\vt')$ but $\vw$ is not $needed(\vA',\vt')$.
Since $\vw_{k-1}$ is $needed(\vA',\vt')$ but not $needed(\vA,\vt)$,
Lemma \ref{order-copies} implies that the \AnnScan\ object $(\vA',\vt')$ is written before $(\vA,\vt)$.
Since \vw\ is not $needed(\vA',\vt')$, it follows from Lemma \ref{order-copies}
that \vw\ is not $needed(\vA,\vt)$.  This contradiction completes the proof of the claim.
\end{proof}

Now we are ready to prove that step $s$ preserves the truth of the lemma.
We need only consider 
steps $s$ that update the \var{cur} variable of a \compact\ routine
or perform a successful CAS at line \ref{CAS-left-SLL}.

Now, suppose $s$ executes line \ref{advance-cur-1} to change the \var{cur} pointer from node \vz\ to \vx.
Then \vx\ was read from \var{\vx->left} at line \ref{init-next}.
By the test at line \ref{cur-too-new}, $\vz.\ts > \vt$.
So, by Claim \ref{needed-claim}, \vx\ is $needed(\vA,\vt)$.

Now, suppose $s$ executes line \ref{advance-cur-2} to change the \var{cur} pointer from node \vz\ to \vx.
Then, \vx\ was read from \var{\vx->left} at line \ref{init-next}.
Since the while loop at line \ref{advance-i} terminated, we have $\vz.\ts > \vA[i]$.
By the test at line \ref{test-next-needed}, we have $\vx\leq\vA[i]$.
So, by Claim \ref{needed-claim}, \vx\ is $needed(\vA,\vt)$.

Now, suppose $s$ performs  a successful CAS at line \ref{CAS-left-SLL}  
that changes \var{x->left} to \var{v}.
Then \vx\ and \var{v} are the values of the local variables \var{cur} and \var{newNext}, respectively, when the CAS occurs.
Since \vx\ is stored in \var{cur}, \vx\ is $needed(\vA,\vt)$, by the induction hypothesis.
Claim \ref{CAS-needed} says that \var{v} is $needed(\vA,\vt)$ and any node \vw\ satisfying 
$\var{v} \prec \vw \prec \vx$ is not $needed(\vA,\vt)$.

Finally, suppose $s$ executes line \ref{advance-cur-3} to change the \var{cur} pointer from node \vz\ to \vw.
Let \vy\ be the node stored in \var{newNext}.
By Claim \ref{CAS-needed}, $\vy \prec \vz$ and \vy\ is $needed(\vA,\vt)$.
We argue that, prior to $s$, a node $\vx \preceq \vy$ was stored in \var{\vz->left} by
considering two cases.
If the loop at lines \ref{CAS-left-SLL}--\ref{test-overshoot} terminated
because the CAS at line \ref{CAS-left-SLL} succeeded, then that CAS stored \vy\ in \var{\vz->left}.
Otherwise, the test at line \ref{test-overshoot} was true, so the node \vx\ read
from \var{\vz->left} on line \ref{reread-left} satisfied
$\vx.\ts \leq \vy.\ts$.
Since the test at line \ref{cur-too-new} failed, $\vz.\ts \leq \vt$.
Since $\vy\prec\vz$, $\vy.\ts \leq \vz.\ts \leq \vt$ by precondition \ref{append-order}.
Since \vy\ is $needed(\vA,\vt)$, it must be the last appended node with its timestamp.
Therefore, $\vx.\ts \leq \vy.\ts$ implies that $\vx \preceq \vy$.
Thus, in either case, a node $\vx \preceq \vy$ was stored in \var{\vz->left} before $s$.
By Lemma \ref{sll-order}, the node \vw\ that $s$ reads in \var{\vz->left} satisfies $\vw \preceq \vx$.
We show that \vw\ is $needed(\vA,\vt)$ by considering two cases.
If $\vw = \vy$, we already know that \vy\ is $needed(\vA,\vt)$.
If $\vw \neq \vy$, we have $\vw \preceq \vx \preceq \vy \prec \vz$.
Assume the CAS that stored \vw\ in \var{x->left} was performed by a \var{\compact($\vA'$,$\vt'$)}.
By the induction hypothesis, \vy\ is not $needed(\vA',\vt')$ since $\vw \prec \vy \prec \vz$
and a CAS in \var{\compact($\vA'$,$\vt'$)} stored \vw\ in \var{\vz->left}.
Since \vy\ is $needed(\vA,\vt)$, Lemma \ref{order-copies} implies that
$(\vA,\vt)$ was written before $(\vA',\vt')$.
Since \vw\ is $needed(\vA',\vt')$, it follows from Lemma \ref{order-copies} that 
\vw\ is $needed(\vA,\vt)$.
\end{proof}

We linearize successful \tryAppend\ operations when they CAS the \head\ pointer of the list.
We linearize a \search\ operation when it reads the \var{head} at line \ref{init-x}.
The following theorem guarantees that every \search\ operation returns the correct
response.

\begin{theorem}
If a \var{search(k)} operation returns the value from  the last node  appended to the list
before the linearization point of the \search\ whose timestamp is at most \vk.
\end{theorem}
\begin{proof}
\here{This proof could use some polish}
The \search\ returns the value from the node stored in the local variable \vx.

If \vx\ is the value read from \var{head} at line \ref{init-x}, then \vx\ is the last node appended to
the list before the linearization point of the \search, and its timestamp is at most \vt\ since
the test at line \ref{test-x} failed.

Otherwise, there is a configuration $C$ during the search
just before \vx\ is read from the \var{left} pointer of some node \vz\ at line \ref{advance-x}.
By the test at line \ref{test-x}, we know that $\var{\vx->\ts} \leq \vt < \var{\vz->\ts}$.
To derive a contradiction, suppose \vx\ is not the last appended node whose timestamp is at most \vk.
Let \vy\ be the last appended node whose timestamp is at most \vk.
Then, $\vx \prec \vy \prec \vz$.

By Lemma \ref{remove-unneeded}, there is an invocation of \var{compact(\vA,\vt)}
before $C$ such that \vy\ is not $needed(\vA,\vt)$.
We consider two cases.

If $\vk\in\vA$, then \vy\ is not the last appended node whose timestamp is at most \vk, since \vy\ is
not $needed(\vA,\vt)$.  This is the desired contradiction.

Otherwise, $\vk\notin\vA$.
At $C$, \vk\ is the in the \var{Announce} array, by precondition \ref{searches-announced}.
By precondition \ref{announce-pre}, this means that $\vk\geq \vt$.
Since \vy\ is not $needed(\vA,\vt)$, $\var{\vy->\ts} \leq \vt$,
and there is a node appended after \vy\ whose timestamp is also at most \vt.
Since $\vk\geq\vt$, there is another node appended after \vy\ whose timestamp is at most \vk.
Again, this is a contradiction.
\end{proof}

\subsection{Analysis}

We now consider progress properties.
We use Invariant \ref{sll-order} to prove wait-freedom and that \search\ does not fall off the end of the list.

\begin{lemma}
The \compact\ and \search\ routines are wait-free.
A \search\ never sets its local variable \vx\ to null.
\end{lemma}

\begin{proof}
We first show that \compact\ is wait-free.
Each iteration of the outer loop updates \var{cur} at line \ref{advance-cur-1}, \ref{advance-cur-2}
or \ref{advance-cur-3} to a value read from \var{cur->left} at line \ref{init-next} or \ref{advance-cur-3}.
Consider the sequence of values $\vx_1, \vx_2, \ldots$ assigned to the \var{cur} 
variable of the \compact\ routine.
By Invariant \ref{sll-order}, all the values are non-null and we have $\vx_1 \succ \vx_2 \succ \cdots$.
Since there are finitely many nodes that precede $\vx_1$ in the order $\prec$, the loop will terminate.

We must also check that the inner loops terminate.
Due to the test at line \ref{test-sentinel}, \var{cur} is not the sentinel node when
the loop at line \ref{advance-i} is executed.  Thus, $\var{cur->\ts}\geq 0$, so when the index
\vi\ reaches 0, $\vA[0] = -1 < \var{cur->\ts}$ and the loop terminates.
(This argument also ensures that we never use an out-of-bounds index for \var{A}.)

Each iteration of the loop at lines \ref{test-new-next-obs}--\ref{advance-new-next}
sets \var{newNext} to an earlier node in the order $\prec$, by Invariant \ref{sll-order}.
So, it must terminate.

Each time the test at line \ref{CAS-left-SLL}  fails,  some other operation
has performed a successful CAS on \var{cur->left}.
By Invariant \ref{sll-order}, each such CAS changes \var{cur->left}
to an earlier node in the order $\prec$, so it must eventually terminate.

We now show that \search\ is wait-free and does not set \vx\ to null.
Line \ref{init-x} initializes
\vx\ to the node read from \var{head}, which has been appended (by definition).
If line \ref{advance-x} changes \var{x}, then the node previously stored in \vx\ was
not the sentinel, since its timestamp was greater than $\vk\geq 0$.
Thus, the value read from its left pointer is non-null, by Lemma \ref{sll-order}.

If the sequences of nodes visited by the \search\ is $\vx_1, \vx_2, \ldots$, then
by Lemma \ref{sll-order}, we have $\vx_1 \succ \vx_2 \succ \cdots$.
Since there are only finitely many nodes that precede $\vx_1$ in the order $\prec$,
the loop must terminate.
\end{proof}

For \DL, we proved Lemmas \ref{no-crossovers} and \ref{remove-only}
to show that CAS steps on \var{left} pointers can only remove nodes; they cannot cause
a remove node to return to the list.
We now prove analogous results for \SL.

\begin{lemma}
\label{no-crossovers-SLL}
Suppose $\vw \prec \vx \prec \vy \prec \vz$ and $\vw \leftarrow \vy$ at some time during the execution.
Then there is never a time when $\vx \leftarrow \vz$.
\end{lemma}
\begin{proof}
To derive a contradiction, suppose the claim is false.

Since there is a node \vx\ between \vw\ and \vy\ in the total order $\prec$, 
the \var{head} pointer takes the value \vx\ between the times it takes the values \vw\ and \vy.
Thus, \vy\ is not added to the list
by a \var{\tryAppend(\vw,\vy)} operation, so the initial value of \var{\vy->left} is not \vw.
So, there is a CAS on line \ref{CAS-left-SLL} of a \var{\compact(\vA,\vt,\vh)} operation that changes \var{\vy->left} to \vw.
Similarly, there is a CAS of a \var{\compact(\vA,\vt,\vh)} operation that changes
\var{\vz->left} to \vx.
By Lemma \ref{remove-unneeded}, \vx\ is $needed(\vA',\vt')$ but not $needed(\vA,\vt)$.
By Lemma \ref{order-copies}, $(\vA',\vt')$ is written before $(\vA,\vt)$.
Since the arguments of \compact\ are obtained by taking a snapshot of \GlobalAnnScan\ and \var{head},
this means that $\vh'$ was stored in \var{head} before $\vh$, so $\vh' \preceq \vh$.
Thus, $\vw \prec \vx \preceq \vh' \preceq \vh$, so $\vw\neq \vh$.
By Lemma \ref{remove-unneeded}, \vw\ is $needed(\vA,\vt)$ but not $needed(\vA',\vt')$.
This contradicts Lemma~\ref{order-copies}.
\end{proof}

\begin{lemma}
\label{remove-only-SLL}
If a CAS at line \ref{CAS-left-SLL} changes \var{\vz->left} from \vy\ to \vw, we have
$\vw \lreach \vy$ in the preceding configuration~$C$.
\end{lemma}
\begin{proof}
Refer to Figure \ref{fig-remove-only}.
Let \vx\ be the minimum node (with respect to $\prec$) such that $\vx \succeq \vw$ and $\vx \lreach \vy$ in $C$.
(This minimum is well-defined, since there exists a node \vx\ that satisfies both of these properties:
$\vy \succeq \vw$ by Lemma \ref{CAS-advances} and $\vy \lreach \vy$ holds trivially.)
Since $\vx \lreach \vy \leftarrow \vz$, we have $\vx \prec \vz$ by Lemma \ref{sll-order}.

Let $\var{v}$ be the value of \var{\vx->left} in $C$.  By Lemma \ref{sll-order},
$\var{v} \prec \vx$.
It must be the case that $\var{v} \prec \vw$; 
otherwise we would have $\var{v} \succeq \vw$ and $\var{v} \lreach \vy$, which would
contradict the fact that \vx\ is the \emph{minimum} node that satisfies these two properties.

To derive a contradiction, suppose $\vw \prec \vx$.  Then we have
$\var{v} \prec \vw \prec \vx \prec \vz$.  
We also have $\vw \leftarrow \vz$ after the CAS and $\var{v} \leftarrow \vx$ in $C$, which contradicts Lemma \ref{no-crossovers-SLL}.
Thus, $\vw \succeq \vx$.  By definition of \vx, $\vx \succeq \vw$ and $\vx\lreach\vy$ in $C$.
Thus, $\vx = \vw$ and $\vw\lreach\vy$.
\end{proof}

Lemma \ref{remove-only-SLL} allows us to prove that the \compact\ routine permanently removes unneeded nodes.

\Eric{Should we strengthen the following proposition to take into account how \vA\ and \vt\ are obtained to argue that they are up-to-date and hence any versions that are unneeded at when we start compacting are removed?}

\begin{proposition}
\label{SLL-progress}
Suppose  a \var{compact(\vA,\vt,\vh)} routine on a version list has terminated prior to some configuration $C$.
Then all nodes \vx\ satisfying  $\mbox{sentinel} \prec \vx\prec\vh$ that 
are reachable from the \var{head} of the version list in $C$
are $needed(\vA,\vt)$.
\end{proposition}
\begin{proof}
To derive a contradiction, assume that some node \vx\ is reachable from \var{head} in $C$
and \vx\ is not $needed(\vA,\vt)$.
Consider the sequence of values $\vh=\vy_1, \vy_2, \ldots, \vy_\ell=\mbox{sentinel}$ that are assigned to the
local variable \var{cur} during the \var{\compact(\vA,\vt,\vh)} routine.
Since each is read from the \var{left} pointer of the previous one, $\vy_\ell\prec \vy_{\ell-1} \prec \cdots \prec \vy_1$.
Since \vx\ is not $needed(\vA,\vt)$, Lemma \ref{remove-unneeded} implies that \vx\ does not appear in this sequence.
So, for some $j$, $\vy_{j+1} \prec \vx \prec \vy_{j}$.

Now, let $\var{head} = \vz_1, \vz_2, \ldots, \vz_k = \vx$ be the path from \var{head} to \vx\ by following \var{left} pointers.
Consider the maximum $i$ such that $\vz_i \succeq \vy_{j}$.
(Such an $i$ exists, since $\vz_1 = \var{head} \succeq \vh = \vy_1 \succeq \vy_{j}$.)
So, we have $\vy_{j+1} \prec \vx = \vz_k \prec \vz_{k-1} \prec \cdots \vz_{i+1} \prec \vy_j \leq \vz_i$.
Now, $\vy_{j+1}\leftarrow \vy_j$ at some time during the \compact\ (before $C$).
If $\vy_j$ were equal to $\vz_i$, then in $C$, $\var{$\vz_i$->left} = \var{$\vy_j$->left} \preceq \vy_{j+1}$ by Invariant \ref{sll-order}, and this would contradict the fact that $\var{$\vz_i$->left} = \vz_{i+1} \succ \vy_{j+1}$.
So, we have $\vy_{j+1} \prec \vz_{i+1} \prec \vy_j \prec \vz_i$.
However, $\vy_{j+1}\leftarrow \vy_j$ at some time before $C$ and $\vz_{i+1} \leftarrow \vz_i$ in $C$.
This contradicts Lemma \ref{no-crossovers-SLL}.
\end{proof}

The head of a version list represents the current version, so it is not obsolete.
Thus, when a node is added to the range tracker data structure,
it is no longer the head of its  list.
Once a process receives a node \vx\ returned by the range tracker, 
it gets up-to-date copies $\vA, \vt$ and $\vh$ of the
announcements, global timestamp and  the list's \var{head} and invoke \var{\compact(\vA,\vt,\vh)}.
Since \vx\ was already returned by the range tracker, \vx\ will not be $needed(\vA,\vt)$
and $\vx \prec \vh$.
Proposition \ref{SLL-progress} ensures that \vx\ is permanently
removed from its version list before the \compact\ terminates.
Thus, aside from the lists
where \compact\ routines are still pending, all nodes
that have been returned by the range tracker are no longer reachable from the head of the list.

\subsection{Reserving Timestamps} 
\label{avoid-copy}
Both \bbf\ and our singly-linked list compaction algorithm from Section~\ref{sec:sll} use atomic copy (a primitive which allows us to atomically read from one memory location and write to another) to announce timestamps in a wait-free manner. Instead, in our experiments, we use a lighter-weight lock-free scheme for announcing timestamps which consist of (A1) reading the current timestamp, (A2) announcing it, and (A3) checking if the current timestamp is still equal to the announced value and, if not, going back to step A1.
This means that a process scanning the announcement array might see a very old timestamp get announced, which breaks the monotonicity property required for the correctness of our singly-linked list.
We fix this by having all processes work on updating a global scan of the announcement array as described in Section~\ref{sec:sll}.
When reading the announcement array, we ignore any timestamp that (1) is smaller than the timestamp at which the current global scan was taken and (2) does not appear in the current global scan.
This is safe because any announced timestamp satisfying these two conditions will fail the check in step A3, and not be used.
This change to how we scan announcement arrays ensures that the global scan satisfies the monotonicity property required by our singly-linked list.
\er{This announcement scanning technique is also used in our implementation of \STEAMLF, \bbf, and \RTSLGC.}
\here{fix names}

\section{Additional Experiments}
\label{additional-experiments}

Figures~\ref{fig:tree-100K-undersub-uniform}--\ref{fig:hashtable-100K-update-heavy-uniform} showcase the same workloads as Figures~\ref{fig:tree-100K-undersub-zipf}--\ref{fig:hashtable-100K-update-heavy}, but with keys drawn from uniform rather than zipfian distribution.

\include{appendix-figures}

%% file: remove-only.pdf_t
\begin{picture}(0,0)%
\includegraphics{remove-only.pdf}%
\end{picture}%
\setlength{\unitlength}{2763sp}%
\begingroup\makeatletter\ifx\SetFigFont\undefined%
\gdef\SetFigFont#1#2#3#4#5{%
  \reset@font\fontsize{#1}{#2pt}%
  \fontfamily{#3}\fontseries{#4}\fontshape{#5}%
  \selectfont}%
\fi\endgroup%
\begin{picture}(4981,1688)(1486,-3309)
\put(4426,-3211){\makebox(0,0)[lb]{\smash{{\SetFigFont{8}{9.6}{\rmdefault}{\mddefault}{\updefault}{\color[rgb]{0,0,0}$\preceq$}%
}}}}
\put(6226,-2236){\makebox(0,0)[lb]{\smash{{\SetFigFont{8}{9.6}{\rmdefault}{\mddefault}{\updefault}{\color[rgb]{0,0,0}\vz}%
}}}}
\put(5176,-2236){\makebox(0,0)[lb]{\smash{{\SetFigFont{8}{9.6}{\rmdefault}{\mddefault}{\updefault}{\color[rgb]{0,0,0}\vy}%
}}}}
\put(1501,-1786){\makebox(0,0)[lb]{\smash{{\SetFigFont{8}{9.6}{\rmdefault}{\mddefault}{\updefault}{\color[rgb]{0,0,0}Before CAS}%
}}}}
\put(1501,-2761){\makebox(0,0)[lb]{\smash{{\SetFigFont{8}{9.6}{\rmdefault}{\mddefault}{\updefault}{\color[rgb]{0,0,0}After CAS}%
}}}}
\put(1576,-3211){\makebox(0,0)[lb]{\smash{{\SetFigFont{8}{9.6}{\rmdefault}{\mddefault}{\updefault}{\color[rgb]{0,0,0}\var{v}}%
}}}}
\put(3676,-3211){\makebox(0,0)[lb]{\smash{{\SetFigFont{8}{9.6}{\rmdefault}{\mddefault}{\updefault}{\color[rgb]{0,0,0}\vx}%
}}}}
\put(2626,-3211){\makebox(0,0)[lb]{\smash{{\SetFigFont{8}{9.6}{\rmdefault}{\mddefault}{\updefault}{\color[rgb]{0,0,0}\vw}%
}}}}
\put(1576,-2236){\makebox(0,0)[lb]{\smash{{\SetFigFont{8}{9.6}{\rmdefault}{\mddefault}{\updefault}{\color[rgb]{0,0,0}\var{v}}%
}}}}
\put(3676,-2236){\makebox(0,0)[lb]{\smash{{\SetFigFont{8}{9.6}{\rmdefault}{\mddefault}{\updefault}{\color[rgb]{0,0,0}\vx}%
}}}}
\put(2626,-2236){\makebox(0,0)[lb]{\smash{{\SetFigFont{8}{9.6}{\rmdefault}{\mddefault}{\updefault}{\color[rgb]{0,0,0}\vw}%
}}}}
\put(2101,-2236){\makebox(0,0)[lb]{\smash{{\SetFigFont{8}{9.6}{\rmdefault}{\mddefault}{\updefault}{\color[rgb]{0,0,0}$\prec$}%
}}}}
\put(2101,-3211){\makebox(0,0)[lb]{\smash{{\SetFigFont{8}{9.6}{\rmdefault}{\mddefault}{\updefault}{\color[rgb]{0,0,0}$\prec$}%
}}}}
\put(3151,-3211){\makebox(0,0)[lb]{\smash{{\SetFigFont{8}{9.6}{\rmdefault}{\mddefault}{\updefault}{\color[rgb]{0,0,0}$\preceq$}%
}}}}
\put(3151,-2236){\makebox(0,0)[lb]{\smash{{\SetFigFont{8}{9.6}{\rmdefault}{\mddefault}{\updefault}{\color[rgb]{0,0,0}$\preceq$}%
}}}}
\put(5176,-3211){\makebox(0,0)[lb]{\smash{{\SetFigFont{8}{9.6}{\rmdefault}{\mddefault}{\updefault}{\color[rgb]{0,0,0}\vy}%
}}}}
\put(5701,-3211){\makebox(0,0)[lb]{\smash{{\SetFigFont{8}{9.6}{\rmdefault}{\mddefault}{\updefault}{\color[rgb]{0,0,0}$\prec$}%
}}}}
\put(6226,-3211){\makebox(0,0)[lb]{\smash{{\SetFigFont{8}{9.6}{\rmdefault}{\mddefault}{\updefault}{\color[rgb]{0,0,0}\vz}%
}}}}
\end{picture}%

%% file: remove-succeeds.pdf_t
\begin{picture}(0,0)%
\includegraphics{remove-succeeds.pdf}%
\end{picture}%
\setlength{\unitlength}{2763sp}%
\begingroup\makeatletter\ifx\SetFigFont\undefined%
\gdef\SetFigFont#1#2#3#4#5{%
  \reset@font\fontsize{#1}{#2pt}%
  \fontfamily{#3}\fontseries{#4}\fontshape{#5}%
  \selectfont}%
\fi\endgroup%
\begin{picture}(5566,944)(1793,-2334)
\put(1876,-2236){\makebox(0,0)[lb]{\smash{{\SetFigFont{8}{9.6}{\rmdefault}{\mddefault}{\updefault}{\color[rgb]{0,0,0}\vw}%
}}}}
\put(6976,-1561){\makebox(0,0)[lb]{\smash{{\SetFigFont{8}{9.6}{\rmdefault}{\mddefault}{\updefault}{\color[rgb]{0,0,0}\var{head}}%
}}}}
\put(4126,-2236){\makebox(0,0)[lb]{\smash{{\SetFigFont{8}{9.6}{\rmdefault}{\mddefault}{\updefault}{\color[rgb]{0,0,0}$\var{v}'$}%
}}}}
\put(4876,-2236){\makebox(0,0)[lb]{\smash{{\SetFigFont{8}{9.6}{\rmdefault}{\mddefault}{\updefault}{\color[rgb]{0,0,0}\vz}%
}}}}
\put(5626,-2236){\makebox(0,0)[lb]{\smash{{\SetFigFont{8}{9.6}{\rmdefault}{\mddefault}{\updefault}{\color[rgb]{0,0,0}\var{v}}%
}}}}
\put(2626,-2236){\makebox(0,0)[lb]{\smash{{\SetFigFont{8}{9.6}{\rmdefault}{\mddefault}{\updefault}{\color[rgb]{0,0,0}\vx}%
}}}}
\end{picture}%

%% file: search-inv.pdf_t
\begin{picture}(0,0)%
\includegraphics{search-inv.pdf}%
\end{picture}%
\setlength{\unitlength}{2644sp}%
\begingroup\makeatletter\ifx\SetFigFont\undefined%
\gdef\SetFigFont#1#2#3#4#5{%
  \reset@font\fontsize{#1}{#2pt}%
  \fontfamily{#3}\fontseries{#4}\fontshape{#5}%
  \selectfont}%
\fi\endgroup%
\begin{picture}(5948,1603)(1411,-2993)
\put(1426,-2911){\makebox(0,0)[lb]{\smash{{\SetFigFont{8}{9.6}{\rmdefault}{\mddefault}{\updefault}{\color[rgb]{0,0,0}$C$:}%
}}}}
\put(6976,-1561){\makebox(0,0)[lb]{\smash{{\SetFigFont{8}{9.6}{\rmdefault}{\mddefault}{\updefault}{\color[rgb]{0,0,0}\var{head}}%
}}}}
\put(1876,-2236){\makebox(0,0)[lb]{\smash{{\SetFigFont{8}{9.6}{\rmdefault}{\mddefault}{\updefault}{\color[rgb]{0,0,0}$\vx'$}%
}}}}
\put(2626,-2236){\makebox(0,0)[lb]{\smash{{\SetFigFont{8}{9.6}{\rmdefault}{\mddefault}{\updefault}{\color[rgb]{0,0,0}\vy}%
}}}}
\put(4126,-2236){\makebox(0,0)[lb]{\smash{{\SetFigFont{8}{9.6}{\rmdefault}{\mddefault}{\updefault}{\color[rgb]{0,0,0}\vz}%
}}}}
\put(4876,-2236){\makebox(0,0)[lb]{\smash{{\SetFigFont{8}{9.6}{\rmdefault}{\mddefault}{\updefault}{\color[rgb]{0,0,0}\vx}%
}}}}
\put(5626,-2236){\makebox(0,0)[lb]{\smash{{\SetFigFont{8}{9.6}{\rmdefault}{\mddefault}{\updefault}{\color[rgb]{0,0,0}$\vz'$}%
}}}}
\put(1876,-2911){\makebox(0,0)[lb]{\smash{{\SetFigFont{8}{9.6}{\rmdefault}{\mddefault}{\updefault}{\color[rgb]{0,0,0}$\vx'$}%
}}}}
\put(4876,-2911){\makebox(0,0)[lb]{\smash{{\SetFigFont{8}{9.6}{\rmdefault}{\mddefault}{\updefault}{\color[rgb]{0,0,0}\vx}%
}}}}
\put(1426,-2236){\makebox(0,0)[lb]{\smash{{\SetFigFont{8}{9.6}{\rmdefault}{\mddefault}{\updefault}{\color[rgb]{0,0,0}$C'$:}%
}}}}
\end{picture}%

%% file: appendix-figures.tex


\renewcommand{\arraystretch}{1.2}
\setlength{\tabcolsep}{3.3pt}

\begin{figure*}
	\begin{subfigure}{0.99\textwidth}
	\begin{tabular}{ l l }
	  \small \hspace{0.2cm} \textbf{Legend for all figures:}	& \includegraphics[width=0.55\linewidth,trim=0 0.5cm 0 0]{graphs/legend.pdf}
	\end{tabular}
	\end{subfigure} \hfill
	\begin{subfigure}{0.74\textwidth}
		\includegraphics[width=0.99\linewidth,trim=0.2cm 0 0.25cm 1.5cm, clip]{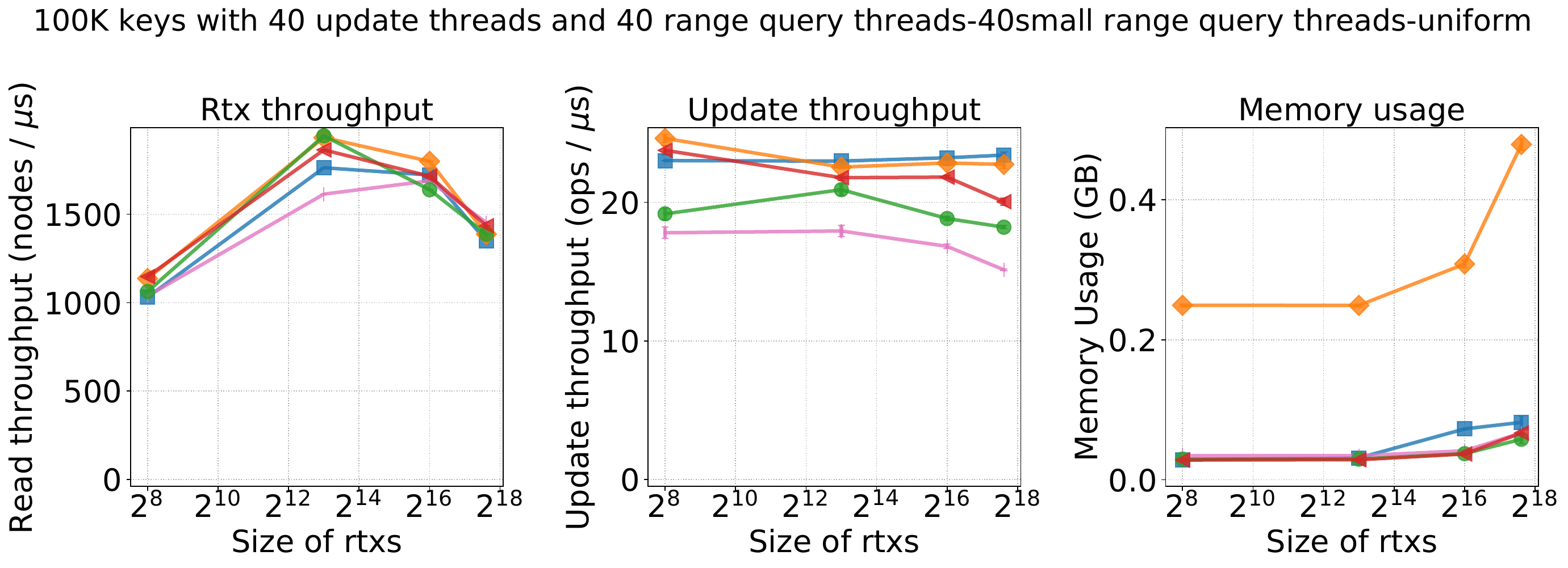}
	\end{subfigure} \hfill
	\begin{subfigure}{0.24\textwidth}
		\footnotesize
		\begin{tabular}{ |c|c|c|c|c| } 
			\multicolumn{5}{c}{\textbf{Average Version List Lengths}} \\
			\hline
			\begin{tabular}{@{}c@{}}\textbf{Size of} \vspace{-0.1cm} \\ \textbf{rtxs}\end{tabular}  & \textbf{$2^{8}$} & \textbf{$2^{13}$} & \textbf{$2^{16}$} & \textbf{$2^{18}$} \\ \hline \hline
			\textbf{SL-RT} & 1.06 & 1.07 & 1.24 & 1.66 \\ \hline
			\textbf{DL-RT} & 1.07 & 1.08 & 1.23 & 1.63 \\ \hline
			\textbf{BBF+} & 1.07 & 1.08  & 1.23 & 1.64 \\ \hline
			\textbf{Steam+LF} & 1.51 & 1.53 & 1.68 & 2.06 \\ \hline
			\textbf{EBR} & 1.02 & 1.04 & 1.33 & 2.0 \\ \hline
		\end{tabular}
	\end{subfigure}
	\caption{Tree with 100K keys, 40 update threads, 40 fixed-size rtx threads, 40 variable-size rtx threads.}
	\label{fig:tree-100K-undersub-uniform}
\end{figure*}

\begin{figure*}
	\begin{subfigure}{0.74\textwidth}
		\includegraphics[width=0.99\linewidth,trim=0.2cm 0 0cm 1.5cm, clip]{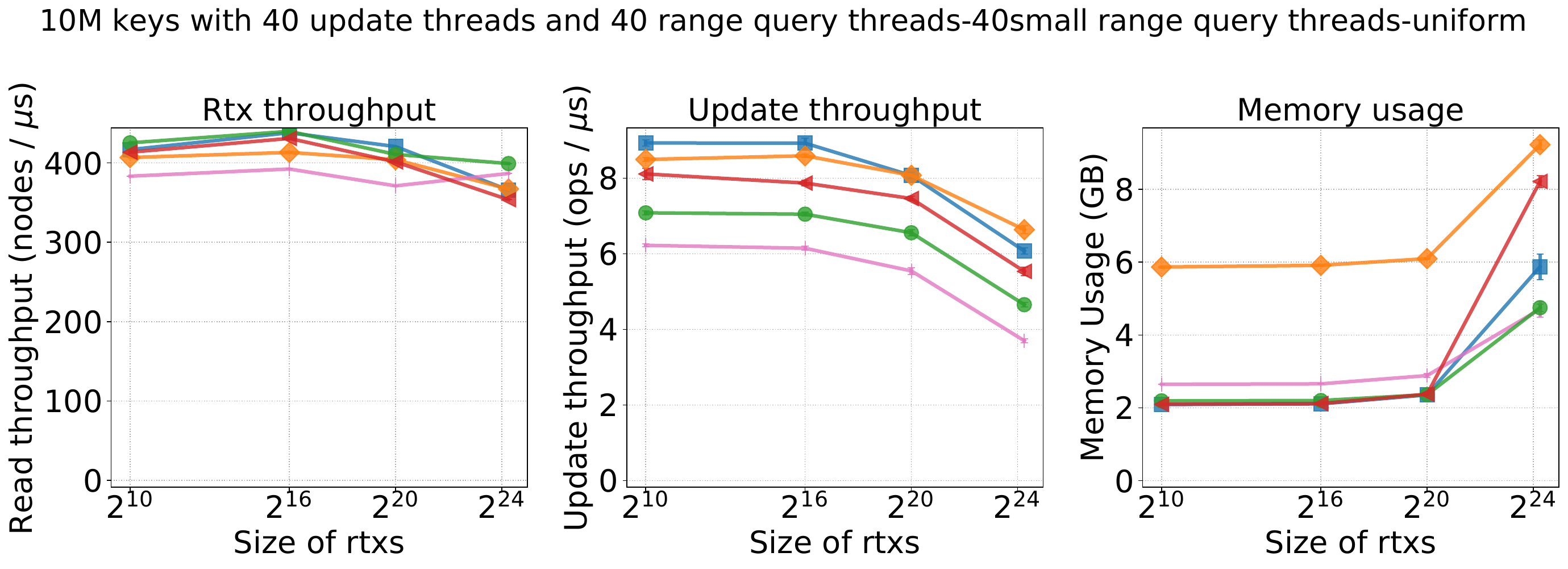}
	\end{subfigure} \hfill
	\begin{subfigure}{0.24\textwidth}
		\footnotesize
		\begin{tabular}{ |c|c|c|c|c| } 
			\multicolumn{5}{c}{\textbf{Average Version List Lengths}} \\
			\hline
			\begin{tabular}{@{}c@{}}\textbf{Size of} \vspace{-0.1cm} \\ \textbf{rtxs}\end{tabular}    & \textbf{$2^{10}$} & \textbf{$2^{16}$} & \textbf{$2^{20}$} & 20M \\ \hline \hline
			\textbf{SL-RT} & 1.0 & 1.01 & 1.09 & 1.63 \\ \hline
			\textbf{DL-RT} & 1.01 & 1.01 & 1.07 & 1.57 \\ \hline
			\textbf{BBF+} & 1.01 & 1.01 & 1.06 & 1.55 \\ \hline
			\textbf{Steam+LF} & 1.43 & 1.43 & 1.46 & 1.68 \\ \hline
			\textbf{EBR} & 1.01 & 1.01 & 1.07 & 1.71 \\ \hline
		\end{tabular}
	\end{subfigure}
	\caption{Tree with 10M keys, 40 update threads, 40 fixed-size rtx threads, 40 variable-size rtx threads.}
	\label{fig:tree-10M-uniform-undersub}
\end{figure*}

\begin{figure*}
	\begin{subfigure}{0.74\textwidth}
		\includegraphics[width=0.99\linewidth,trim=1.99cm 0 1.2cm 1.5cm, clip]{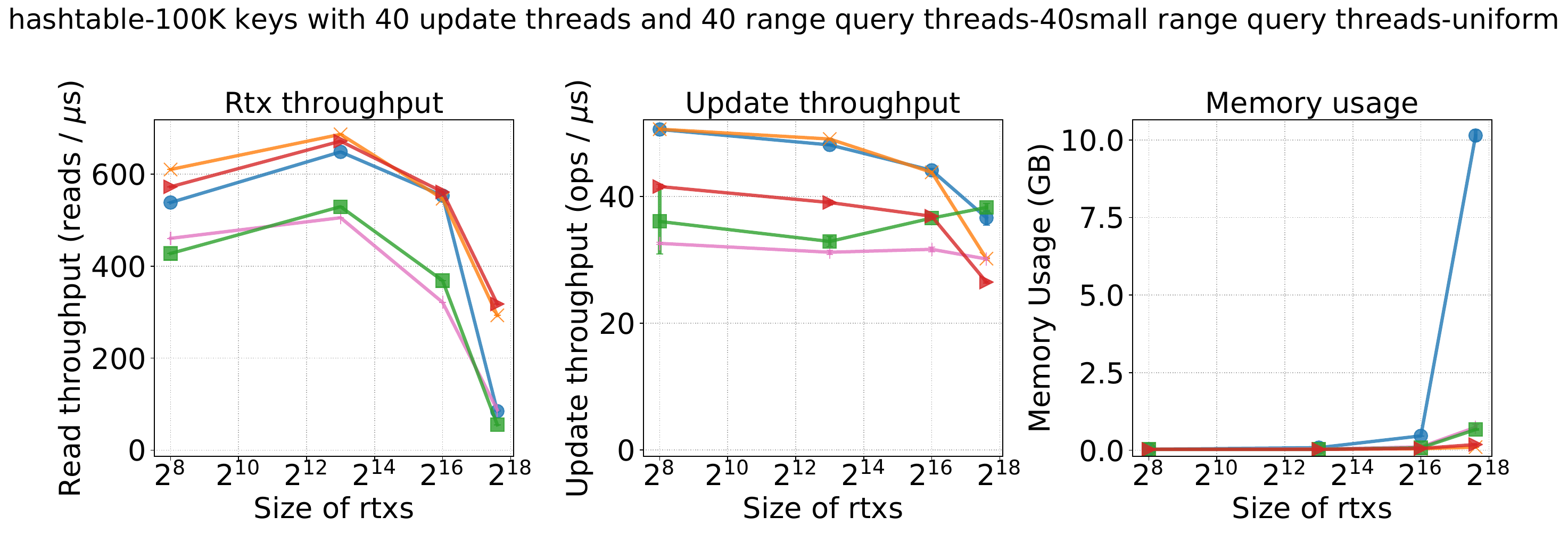}
	\end{subfigure} \hfill
	\begin{subfigure}{0.25\textwidth}
			\footnotesize
			\begin{tabular}{ |c|c|c|c|c| } 
				\multicolumn{5}{c}{\textbf{Average Version List Lengths}} \\
				\hline
				\begin{tabular}{@{}c@{}}\textbf{Size of} \vspace{-0.1cm} \\ \textbf{rtxs}\end{tabular}    & \textbf{$2^{8}$} & \textbf{$2^{13}$} & \textbf{$2^{16}$} & \textbf{$2^{18}$} \\ \hline \hline
				\textbf{SL-RT} & 1.04 & 1.15 & 2.35 & 5.46 \\ \hline
				\textbf{DL-RT} & 1.05 & 1.14 & 2.42 & 25.27 \\ \hline
				\textbf{BBF+} & 1.05 & 1.12  & 2.5 & 19.39 \\ \hline
				\textbf{Steam+LF} & 1.64 & 1.76 & 2.74 & 5.57 \\ \hline
				\textbf{EBR} & 1.01 & 1.19 & 3.34 & 451.99 \\ \hline
			\end{tabular}
	\end{subfigure}
	\caption{Hash table with 100K keys, 40 update threads, 40 fixed-size rtx threads, 40 variable-size rtx threads.}
	\label{fig:hashtable-100K-undersub-uniform}
\end{figure*}

\begin{figure*}
   \begin{subfigure}{0.49\textwidth}
		\centering
		\includegraphics[width=0.99\linewidth,trim=0.20cm 0 12.3cm 1.5cm, clip]{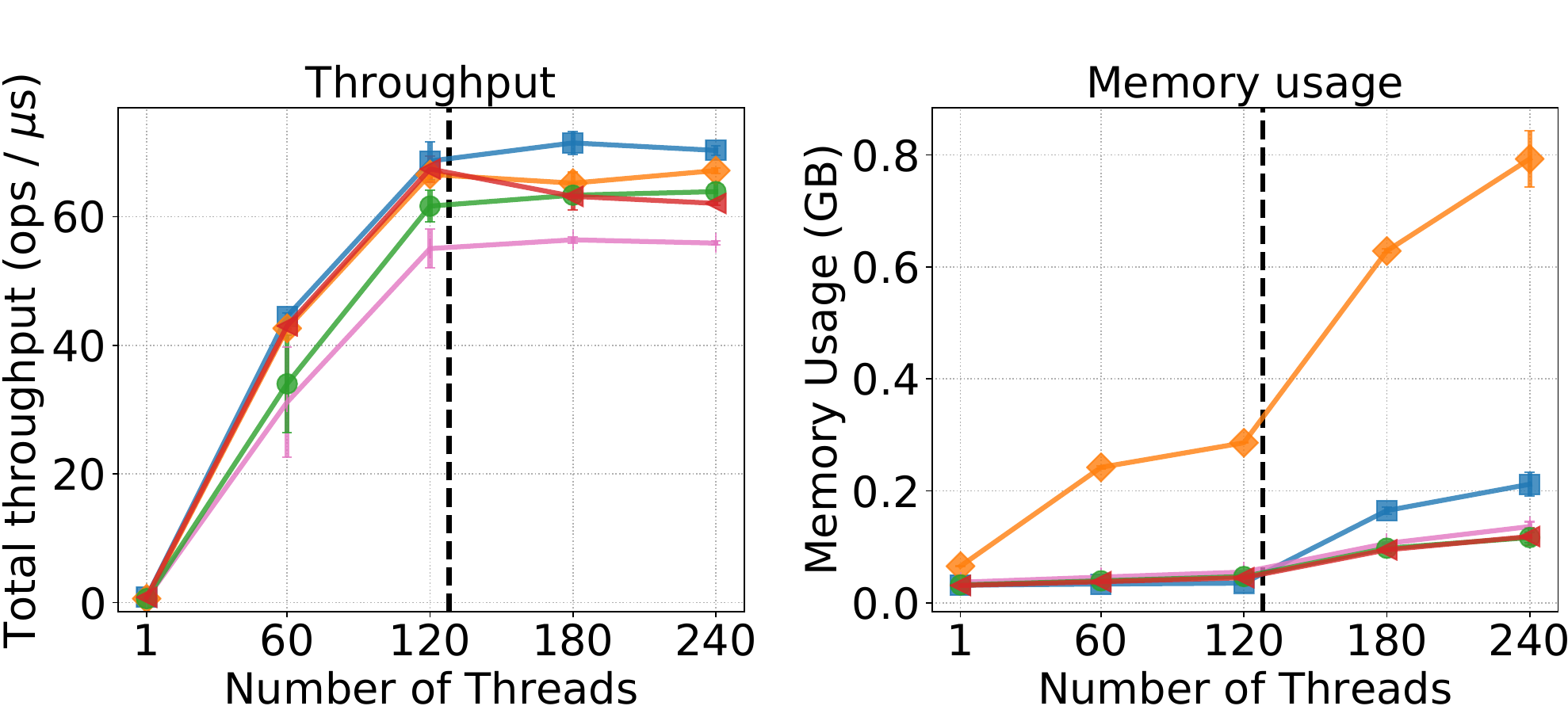}
		\caption{Tree with 100K keys}\label{fig:tree-100K-update-heavy-uniform}
	\end{subfigure} \hfill
	\begin{subfigure}{0.49\textwidth}
		\centering
		\includegraphics[width=0.99\linewidth,trim=0.20cm 0 12.3cm 1.5cm, clip]{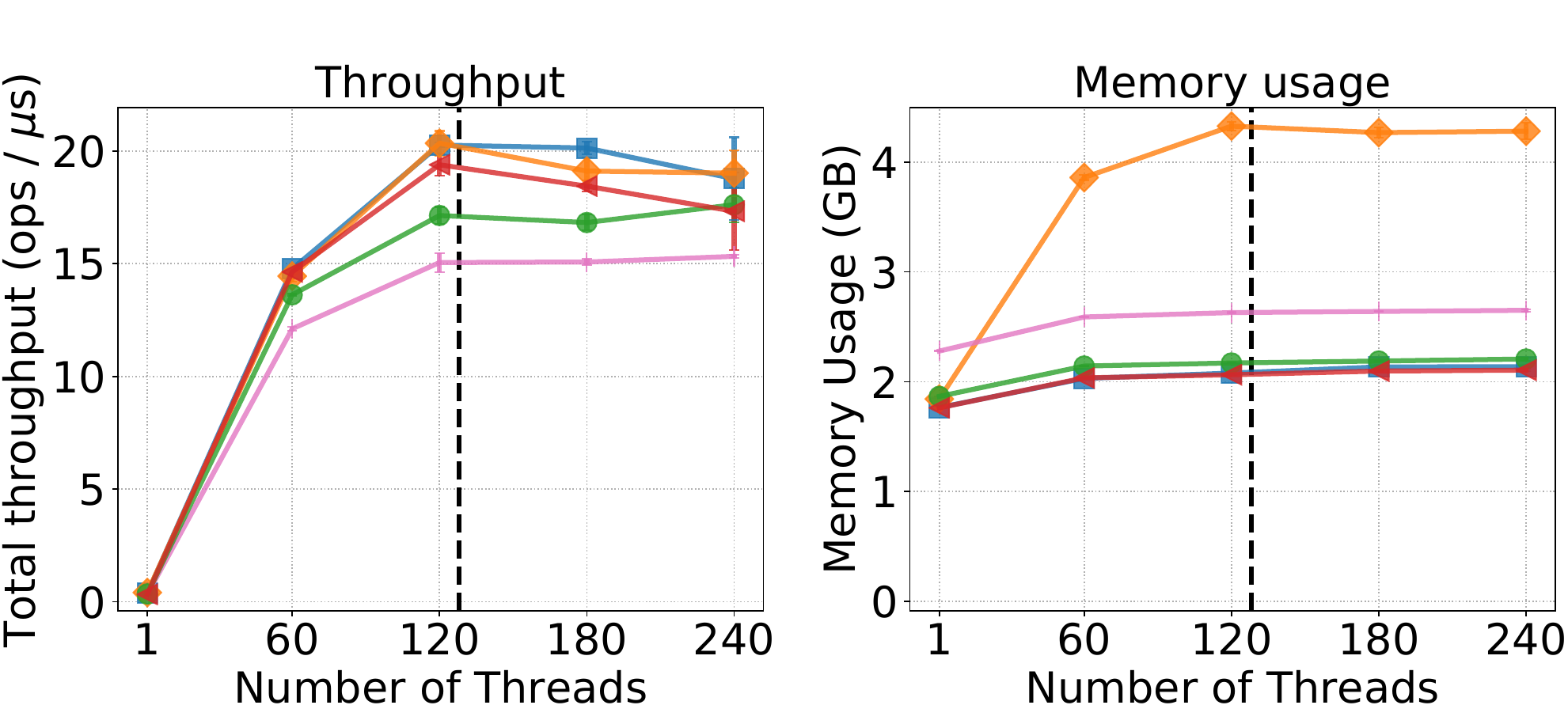}
		\caption{Tree with 10M keys}
		\label{fig:tree-10M-update-heavy-uniform}
	\end{subfigure}	
\caption{Workload with each thread performing 50\% updates, 49\% lookups, and 1\% read transactions of size 1024.}
\label{fig:tree-update-heavy-uniform}
\end{figure*}


\begin{figure}
	\includegraphics[width=0.99\linewidth,trim=0.25cm 0 12.45cm 2cm, clip]{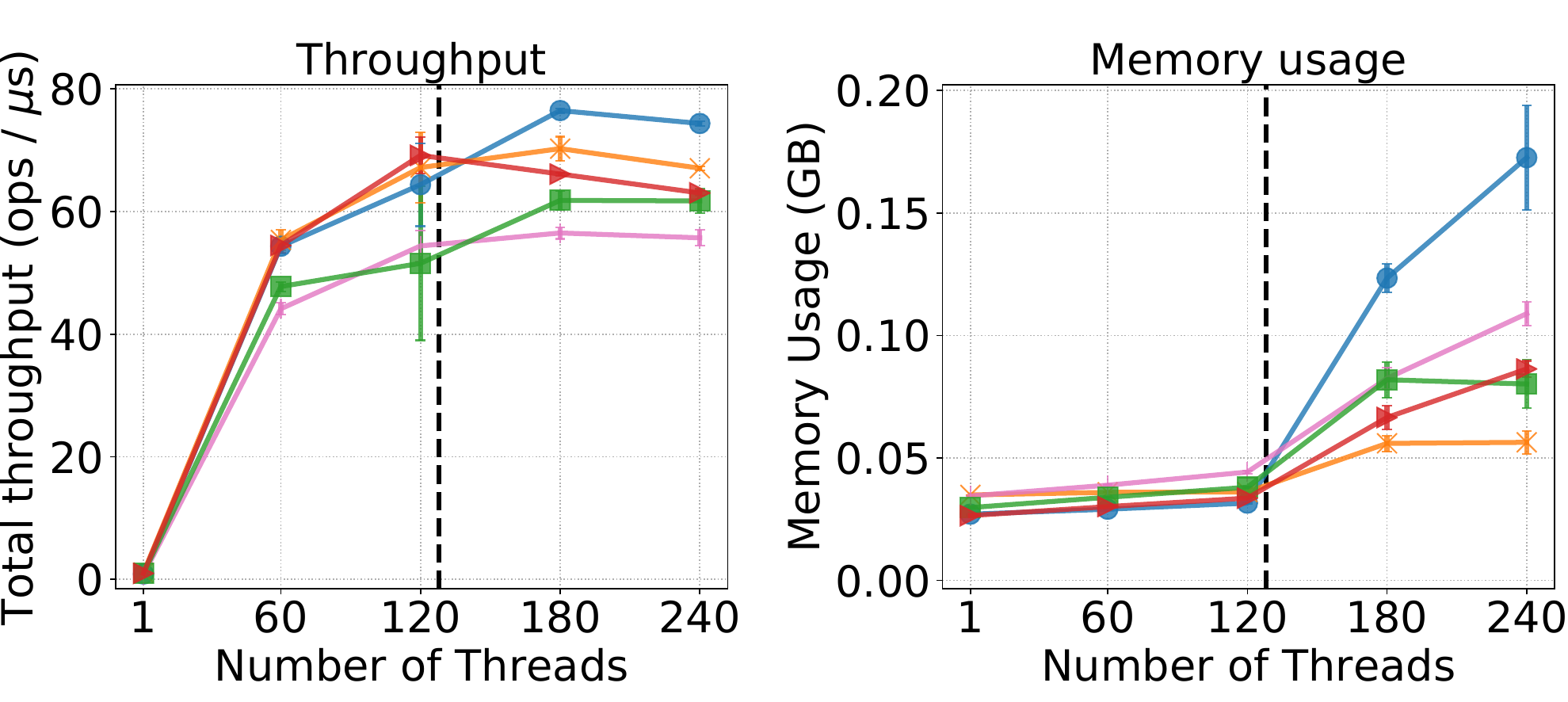}
	\caption{Hash table with 100K keys, each thread performs 50\% updates, 49\% lookups, and 1\% rtxs of size 1024.}
	\label{fig:hashtable-100K-update-heavy-uniform}
\end{figure}